\documentclass[11pt]{article}

\usepackage[T1]{fontenc}
\usepackage{lmodern}
\usepackage{microtype}
\usepackage[margin=1in]{geometry}
\usepackage{amsmath,amssymb,amsthm,mathtools}
\usepackage{enumitem}
\usepackage{hyperref}
\usepackage{xcolor}
\usepackage{cleveref}

\hypersetup{colorlinks=true,linkcolor=blue!50!black,urlcolor=blue!50!black,citecolor=blue!50!black}
\setlist{nosep}
\numberwithin{equation}{section}

\newtheorem{theorem}{Theorem}[section]
\newtheorem{lemma}[theorem]{Lemma}
\newtheorem{corollary}[theorem]{Corollary}
\newtheorem{proposition}[theorem]{Proposition}
\theoremstyle{remark}
\newtheorem{remark}[theorem]{Remark}

\usepackage[]{color-edits}

\addauthor{MM}{teal}

\addauthor{MF}{magenta}

\addauthor{PD}{blue}

\newcommand{\R}{\mathbb{R}}
\newcommand{\E}{\mathbb{E}}
\newcommand{\Pp}{\mathbb{P}}
\newcommand{\ind}{\mathbf{1}}
\newcommand{\Pois}{\operatorname{Pois}}
\newcommand{\Exp}{\operatorname{Exp}}
\newcommand{\OPT}{\mathsf{OPT}}
\newcommand{\ALG}{\mathsf{ALG}}
\newcommand{\OPTP}{\mathsf{OPT}^{\mathrm{Pois}}}
\newcommand{\OPTE}{\mathsf{OPT}^{\mathrm{Exp}}}

\newcommand{\dP}{\delta^{\mathrm{Pois}}_{\mathbf N}}
\newcommand{\dE}{\delta^{\mathrm{Exp}}_Z}
\newcommand{\calP}{\mathcal P_{\lambda}}

\title{Optimal Online Stochastic Matching}
\author{Paul D\"utting\thanks{Google Research, Zurich, Switzerland. Email: \texttt{duetting@google.com}} \and Michal Feldman\thanks{Tel Aviv University, Tel Aviv, Israel. Email: \texttt{mfeldman@tauex.tau.ac.il }} \and Mathieu Molina\thanks{Tel Aviv University, Tel Aviv, Israel. Email: \texttt{mathieu.molina.research@gmail.com}}}
\date{}

\begin{document}
\maketitle

\begin{abstract}

We resolve the problem of unweighted  online stochastic matching with a known arrival law by establishing its optimal competitive ratio. 
For Poisson arrivals, we achieve $1-e/e^e\approx0.8206$, matching the upper bound of Chierichetti, Giacchini, Panconesi, and  Vattani~\cite{ChierichettiEtAl2025}. For exactly $n$ i.i.d.~arrivals, we obtain the same guarantee up to an additive $O(n^{-1/2})$ loss.

Our key technical contribution is an exact correspondence between Poisson matching and exponential scheduling, preserving both online performance and the expected offline optimum. To prove the latter, we show that the distribution of the maximum of a random supermodular set function is determined by its joint distributions along chains of nested sets, a result of independent interest. This reduces the analysis to chains, where classical Poisson–exponential duality applies. The offline equivalence extends to arbitrary nonnegative edge weights. We then design a posterior allocation policy for exponential scheduling and use the correspondence to transfer its optimal competitive guarantee to online matching.

Additionally, we show that the optimal competitive ratio under an unknown arrival law is strictly lower than $1-e/e^e$, establishing a strict separation between the known and unknown arrival models.

\end{abstract}

\section{Introduction} 

Online bipartite matching is a fundamental problem in online algorithms.
A set of $m$ offline vertices is known in advance, while online vertices
arrive sequentially and reveal their neighborhoods. Each arrival must
be either rejected or irrevocably matched to an unmatched neighbor.
The objective is to maximize the size of the matching. In their seminal
work, Karp, Vazirani, and Vazirani~\cite{KVV1990} introduced the
\textsc{Ranking} algorithm and proved that it achieves
the optimal competitive ratio of $1-1/e$ under
adversarial arrivals.

In many applications, the arrival sequence follows a distribution
that can be estimated from historical data. This motivates
the study of {\em online stochastic matching}, in which
the arrival distribution is known in advance. Two popular models are: 
(i) A fixed number of $n$ i.i.d.~arrivals drawn from a known
distribution, and (ii) Poisson arrivals with known rates.
As the expected number of arrivals grows ($n \rightarrow \infty$), the offline
benchmarks in the two models agree up to vanishing relative error,
and online policies for the Poisson model can be transferred to the i.i.d.~model with
vanishing loss relative to the offline optimum~\cite{HuangShu2021,Yan2026JailletLu}.

In a landmark result, Feldman, Mehta, Mirrokni, and
Muthukrishnan~\cite{FeldmanEtAl2009} showed that
knowledge of the arrival distribution (specifically the i.i.d.~model) makes it
possible to achieve a competitive ratio of $0.67$ surpassing the adversarial bound of $1-1/e$, while also demonstrating that the optimal competitive ratio remains bounded away from $1$.

Determining the optimal competitive ratio for
online stochastic matching has since been the subject of extensive research.
A sequence of algorithmic improvements
~\cite{BahmaniKapralov2010,ManshadiEtAl2012,JailletLu2014,HuangShu2021}
culminated in the best known guarantee of $0.716$,
due to Huang, Shu, and Yan~\cite{HuangShuYan2022}, established
in the Poisson arrival model.
In parallel, a series of impossibility
results~\cite{BahmaniKapralov2010,ManshadiEtAl2012}
progressively narrowed the range of possible competitive ratios. 
Most recently, Chierichetti, Giacchini, Panconesi, and
Vattani~\cite{ChierichettiEtAl2025} established an upper bound of
$1-e/e^e\approx0.8206$,
originally for i.i.d.~arrivals, but also applicable
to the Poisson model.
Despite more than a decade of progress, 
a gap of over $0.1$ remained 
between the best
known algorithmic guarantee and the impossibility bound. \medskip


In this paper we resolve this longstanding open
problem by establishing the optimal competitive ratio for
unweighted online stochastic matching. We establish the guarantee in the Poisson arrival model, in which arrivals of each neighborhood $S\subseteq[m]$ follow an independent Poisson process of known
rate $\lambda_S$ over the unit time horizon. Let $\ALG$ and $\OPT$ denote the expected matching sizes of an online algorithm and the optimal offline algorithm, respectively. An online algorithm is $c$-competitive if $\ALG\geq c  \cdot  \OPT$ for every instance. We define the optimal competitive ratio as
\begin{equation*}
\mathsf{CR}\coloneqq
\inf_{(\lambda_S)_{S \subset \mathbb{N}}\geq0}\sup_{\ALG}
\frac{\ALG}{\OPT},
\end{equation*}
where the supremum ranges over online algorithms that know the arrival rates, and the infimum ranges over finite instances with $\OPT>0$. We construct an online policy whose competitive ratio matches the $1-e/e^e$ upper bound of Chierichetti, Giacchini, Panconesi, and  Vattani~\cite{ChierichettiEtAl2025}.

\begin{theorem}\label{thm:CR-poisson}
For Poisson arrivals, we have the tight competitive ratio:
\begin{equation*}
\mathsf{CR}=1-\frac{e}{e^e}. 
\end{equation*}
\end{theorem}
This result is information-theoretic, as the algorithm establishing the guarantee is not polynomial-time. For exactly $n$ i.i.d.~arrivals, we obtain a lower bound on the competitive ratio of $1-e/e^e-O(1/\sqrt{n})$, which is asymptotically optimal. 

\subsection{Our Techniques}

We first establish a Poisson--exponential duality that preserves both online performance and the expected offline optimum, by comparing the matching problem to a scheduling problem. We then design and analyze an online scheduling policy whose competitive guarantee transfers to matching through this duality.

For the special case of exchangeable instances, a direct ODE comparison already yields the tight constant. We present this argument in \Cref{subsec:exchangeable-warmup} as an alternative much shorter proof to the general analysis and as a warm-up that motivates the scheduling comparison policy for the general case.

\paragraph{Poisson--exponential duality.}
A common approach to stochastic matching is to compare online performance with a fluid relaxation of the offline benchmark. The simplest relaxation replaces random arrival counts by their expectations. This already loses too much for our purposes: with a single offline vertex and a Poisson arrival count of mean one, the relaxation has value $1$, whereas the expected maximum matching size is only $1-1/e$. Stronger relaxations, including the Natural LP of Huang and Shu~\cite{HuangShu2021}, incorporate stochastic availability constraints and support guarantees beyond $1-1/e$.

Instead of using an offline relaxation, we prove that the Poisson matching problem is equivalent to an exponential scheduling problem with appropriate service packing constraints, as in \cite{ImKulkarniMunagala2014}, but with a total work objective instead. Our approach therefore gives an exact representation of the expected offline optimum. We replace random arrival counts by mean deterministic processing capacities and unit vertex capacities by independent exponential job sizes. Each offline vertex becomes a job with demand $Z_i\sim\Exp(1)$. Each neighborhood $S$ becomes a server of capacity $\lambda_S$, which the scheduler can share among the jobs in $S$. The scheduler observes completed jobs and the service already provided, but does not know unfinished job sizes. Its objective is to maximize the total work processed by time $1$, including work on unfinished jobs.

Exponential memorylessness allows us to translate online policies between the two models: assigning arrivals to unmatched vertices corresponds to allocating service to unfinished jobs.
This translation preserves expected reward, since the expected work processed on an exponential job equals its probability of completion. Our main structural result shows that the offline benchmarks also agree. 

\medskip
\noindent\textbf{Theorem (Poisson--exponential duality, informal)} [see~\Cref{thm:offline-duality}]
The Poisson matching instance and its exponential scheduling counterpart have the same expected offline optimum.
\medskip

For a single vertex with arrival rate $\lambda$, both expected offline values equal $1-e^{-\lambda}$. This is the probability of at least one arrival in the matching model and the expected minimum of $\lambda$ and a unit-exponential demand in the scheduling model.

To extend this identity to arbitrary overlapping neighborhoods, we must account for competing demands on shared resources. We express both offline optima through supermodular deficiency functions, which is a known generalization of Hall's characterization of the matching deficit. We then use Poisson count and waiting-time duality to relate these deficiencies along chains of nested vertex sets. Finally, we combine the chain comparisons through an Euler-characteristic identity to prove equality of the expected offline benchmarks.

\paragraph{An online scheduling guarantee.}
We analyze scheduling under a general compact, convex, downward-closed region $P$ of feasible service rates. Our scheduling result therefore applies beyond the regions arising from matching.

Let $Y(Z)$ be an offline allocation, with $Y(Z)\in P$ and $0\leq Y(Z)\leq Z$. We construct an online policy for our service rates such that it follows its own induced conditional expectation:  for each unfinished job $i$, it chooses the service rate $s_i(t)=\E[Y_i(Z)\mid\mathcal F_t]$,
where $\mathcal F_t$ is the observed service and job completion history.  
Convexity and downward closure ensure that the resulting service rates are feasible. We prove the following competitive guarantee for maximizing total work in the online scheduling model.

\medskip
\noindent\textbf{Theorem (Exponential scheduling, informal)} [see~\Cref{prop:online-offline-comparison}]
For every feasible offline allocation $Y$ and every job $i$, the posterior policy satisfies
\begin{equation*}
\E[x_i(1)]
\geq
\left(1-\frac{e}{e^e}\right)\E[Y_i(Z)],
\end{equation*}
where $x_i(t)\coloneqq \int_0^t s_i(u) \, du$ is the work processed on job $i$ by time $t$.
\medskip

We use memorylessness to connect the online service rate to the demand left unallocated offline. While job $i$ is unfinished, we show that
\begin{equation*}
x_i(t)+1-s_i(t)
=
\E[Z_i-Y_i(Z)\mid\mathcal F_t]
\geq0.
\end{equation*}
Since $x_i'(t)=s_i(t)$, cumulative work satisfies the differential inequality $x_i'(t)\leq x_i(t)+1$. We therefore compare our policy with a one-job policy following the associated ODE $x'=x+1$ (which needs not be feasible), stopping at completion. Starting from zero, the ODE induces $e^t-1$ units of service by time $t$. After accounting for completion, this comparison policy processes expected work $1-e/e^e$ by the deadline, which is the probability that the job completes, namely the demand $\Exp(1)$ being at most the total uninterrupted work $e-1$. The comparison policy is separable, as the work processed on each job depends only on its own demand, which allows us to bound the online work one job at a time.

We use the above identity to bound the loss in expected work relative to the comparison policy in terms of the demand left unallocated offline. To do so, we introduce a potential consisting of work already processed plus the expected additional work under the comparison policy. We bound its expected rate of change using the conditional gap and integrate using Dynkin's formula to obtain the scheduling guarantee. Choosing $Y$ to be an optimal offline allocation and applying the Poisson--exponential duality then proves the matching result.

\subsection{Additional Results}

We complement our main guarantee with results on weighted objectives, the value of knowing the arrival law, and the transfer to exactly $n$ arrivals.

\paragraph{Edge-weighted matching.}
We extend the offline duality between matching and scheduling to arbitrary nonnegative edge weights. If matching a type-$S$ arrival to vertex $i$ earns reward $w_{i,S}$, the corresponding server earns the same reward per unit of work processed on job $i$.

\medskip
\noindent\textbf{Theorem (Weighted Poisson--exponential duality, informal)} [see~\Cref{thm:weighted-static}]
For arbitrary nonnegative edge weights, the Poisson matching instance and its exponential scheduling counterpart have the same expected offline optimum.
\medskip

The equality also extends to random edge weights revealed upon arrival. Conditional on its type, each arrival independently draws a weight vector with finite first moments, whose coordinates may be arbitrarily correlated. We obtain this extension by treating the entire weight vector as part of the arrival type.

The same competitive ratio of $1-e/e^e$ that we establish for the unweighted case cannot extend to arbitrary edge weights, since Huang, Shu, and Yan~\cite{HuangShuYan2022} establish an upper bound below $0.703$ for edge-weighted stochastic matching without free disposal. Our coordinatewise scheduling guarantee controls the total work allocated to each job, whereas general edge rewards also depend on which server provides that work.
We nevertheless believe this exact representation can provide a route to competitive guarantees for weighted matching models through the design and analysis of policies for the corresponding scheduling problems.

\paragraph{The value of knowing the arrival law.}
In the unknown-law model, arrivals still follow a fixed stochastic law, but the algorithm is not told this law. Karande, Mehta, and Tripathi~\cite{KarandeMehtaTripathi2011} showed that \textsc{Ranking} beats $1-1/e$ in this setting\footnote{In fact, they provide a guarantee in the random order arrival setting, of which unknown i.i.d. is a special case.}, and Mahdian and Yan~\cite{MahdianYan2011} established a guarantee of $0.696$. The upper bound of Chierichetti, Giacchini, Panconesi, and  Vattani~\cite{ChierichettiEtAl2025} applies to both known and unknown laws. It remained open whether knowing the arrival law strictly improves the optimal competitive ratio. We resolve this question by proving a strict separation between the two models.

\medskip
\noindent\textbf{Theorem (Unknown arrival laws)} [see~\Cref{thm:unknown-poisson}]
There exists a constant $\varepsilon>0$ such that 
\begin{equation*}
\mathsf{CR}^{\text{unknown}} \coloneqq  \sup_{\ALG}  \inf_{(\lambda_S)_{S \subset \mathbb{N}}\geq0} 
\frac{\ALG}{\OPT} \leq 1- \frac{e}{e^e}-\varepsilon < 1- \frac{e}{e^e}=
\inf_{(\lambda_S)_{S \subset \mathbb{N}}\geq0} \sup_{\ALG}
\frac{\ALG}{\OPT} =\mathsf{CR}^{\text{known}}.
\end{equation*}
\medskip

Our construction hides a subset of offline vertices. The offline optimum matches almost all vertices, but an online algorithm cannot learn the hidden subset quickly enough to avoid an additional loss. The separation holds even when the online algorithm knows the prior over arrival laws.

The proof establishes a qualitative separation. Numerical optimization within our hard family (or neighboring ones) gives an upper-bound estimate of approximately $0.816$.\footnote{This numerical exploration was done via a Large Language Model.} This remains well above the previous $0.716$ guarantee for known laws, highlighting why establishing the separation also required a substantial improvement in the known-law lower bound.

\paragraph{Transfer to exactly $n$ arrivals.}
We also sharpen the transfer from Poisson arrivals to exactly $n$ i.i.d.~arrivals. Huang and Shu~\cite{HuangShu2021} obtained an $O(n^{-1/2})$ loss under an approximate-monotonicity assumption on the online policy. Yan~\cite[Appendix~A.1]{Yan2026JailletLu} removed this assumption, albeit with a larger loss of $O(n^{-1/3})$. A small refinement of Yan's simulation, together with concavity of the expected offline matching size in the number of arrivals, gives the stronger rate for arbitrary policies, see \Cref{cor:fixed-count-guarantee}.

\subsection{Further Related Work}

We refer to Huang, Tang, and Wajc~\cite{HuangTangWajc2024} for a broader survey of online matching.


\paragraph{Offline allocations and posterior policies.}
Using information about offline optimal solutions to guide online decisions is a recurring approach in stochastic allocation. Manshadi, Oveis Gharan, and Saberi~\cite{ManshadiEtAl2012} use offline matching statistics, and Tang, Wu, and Wu~\cite{TangWuWu2022} develop algorithms based on finer conditional statistics of offline optimal matchings. Kesselheim, Molinaro, Patton, and Singla~\cite{KesselheimEtAl2026} develop a general posterior-matching framework based on conditional expectations of offline allocations and martingale inequalities, and show that they can obtain a $1-1/e$ competitive guarantee. Our scheduling policy follows this posterior-allocation principle. In our setting, the service decisions themselves determine the information revealed about job sizes. We exploit exponential memorylessness to identify the resulting posterior distributions and analyze the policy through comparison with a one-job service process, obtaining the coordinatewise factor $1-\frac{e}{e^e}$.

\paragraph{Weighted matching and other arrival models.}
Weighted objectives have also received substantial attention. Haeupler, Mirrokni, and Zadimoghaddam~\cite{HaeuplerEtAl2011} and Brubach et al.~\cite{BrubachEtAl2020} study edge-weighted stochastic matching under integral arrival rates. Without this assumption, Yan~\cite{Yan2026JailletLu} and Qiu, Feng,  Zhou, and Wu~\cite{QiuFengZhouWu2023} obtain guarantees exceeding $1-1/e$. Our offline representation applies to arbitrary nonnegative edge weights. Beyond identically distributed arrivals, Tang, Wu, and Wu~\cite{TangWuWu2022} and Chen, Huang, and Sun~\cite{ChenHuangSun2024} study matching with independent, nonidentically distributed arrivals from known distributions.

\paragraph{Non-clairvoyant scheduling.}
Our feasible-rate formulation is related to non-clairvoyant scheduling under packing constraints. Im, Kulkarni, and Munagala~\cite{ImKulkarniMunagala2014} introduce a general framework for scheduling under polyhedral constraints, further developed by J{\"a}ger, Lindermayr, and Megow~\cite{JagerLindermayrMegow2026}. These works study objectives based on job completion times. Our model uses independent exponential job sizes and maximizes the work processed before a common deadline, including partial work on unfinished jobs. This is an early-work objective, studied in other scheduling settings by Sterna and Czerniachowska~\cite{SternaCzerniachowska2017} and Jiang et al.~\cite{JiangEtAl2024}. Our analysis applies to arbitrary compact, convex, downward-closed feasible-rate regions.

\paragraph{Fairness and expected allocations.}
Ma, Xu, and Xu~\cite{MaXuXu2023} study individual and group fairness among offline vertices under known i.i.d.\ arrivals with integral arrival rates, obtaining competitive ratios of $0.725$ and $0.719$, respectively, through LP-based sampling algorithms. Our equality of expected allocation regions and coordinatewise scheduling guarantee imply a competitive ratio of $1-\frac{e}{e^e}$ for both objectives under Poisson arrivals, and $1-\frac{e}{e^e}-O(n^{-1/2})$ for exactly $n$ i.i.d.\ arrivals. These guarantees hold without integrality assumptions and are measured against the optimal offline ex-ante fairness value, but we do not establish a polynomial-time implementation.

\section{From Poisson Matching to Exponential Scheduling}

Our analysis relates online bipartite matching with Poisson arrivals to scheduling with exponential service demands. We first introduce the two models and the polymatroid describing feasible allocations from the mean type supplies, which serves as the feasible-rate region for scheduling. Using a flow decomposition, we then establish a correspondence between online policies in the two models that preserves expected rewards.

\subsection{Matching with Poisson Arrivals:} 

We consider an online bipartite matching model with offline vertex set $[m]\coloneqq\{1,\dots,m\}$. Each online vertex has a type $S\subseteq[m]$, which specifies its neighborhood among the offline vertices. There are therefore $2^m$ possible types.

For each type $S\subseteq[m]$, online vertices arrive according to a Poisson point process of intensity $\lambda_S\geq 0$ on the time interval $[0,T]$, where $T>0$ is the deadline. These processes are independent across types, so the arrival law is fully specified by the intensity vector $\boldsymbol\lambda=(\lambda_S)_{S\subseteq[m]}$, which is known to the decision maker. By rescaling time and multiplying all intensities by $T$, we may assume without loss of generality that the time horizon is $[0,1]$. Under this normalization, the total number $N_S$ of arrivals of type $S$ has distribution $\Pois(\lambda_S)$, independently across types.

For any set $A\subseteq[m]$ of offline vertices, let $N(A)$ denote the number of arrivals with at least one neighbor in $A$, and let $\Lambda(A)$ denote their total arrival intensity:
\begin{equation*}
\Lambda(A)\coloneqq\sum_{S:\,S\cap A\neq\varnothing}\lambda_S,
\qquad
N(A)\coloneqq\sum_{S:\,S\cap A\neq\varnothing}N_S.
\end{equation*}
Then $N(A)\sim\Pois(\Lambda(A))$. We write $\mathbf N=(N_S)_{S\subseteq[m]}$ for the vector of realized type counts. For a vector $z\in\R_+^m$, we use the notation $z(A)=\sum_{i\in A}z_i$ for the associated modular set function.

Upon each arrival, an online policy either rejects the arriving vertex or irrevocably matches it to an unmatched neighbor. Its decision may depend on the current arrival, the history of previous arrivals and decisions, and internal randomization independent of the arrival processes. An offline algorithm is allowed to select a matching after all nodes have arrived. The offline maximum matching size is given by the linear program
\begin{equation*}
\OPTP(\mathbf{N})\coloneqq
\max_{\substack{x_{i,S}\in[0,1]\\S\subseteq[m],\,i\in S}} \left\{
\sum_{S\subseteq[m]}\sum_{i\in S}x_{i,S}
:\;
\sum_{i\in S}x_{i,S}\leq N_S \quad  \forall S\subseteq[m], \quad 
\sum_{S\ni i}x_{i,S}\leq 1 \quad  \forall i\in[m]
\right\}.
\end{equation*}
Here $x_{i,S}$ represents the number of type-$S$ arrivals matched to offline vertex $i$. The constraints ensure that assignments do not exceed the realized supply of each type and that each offline vertex is matched at most once. Integrality of the bipartite matching polytope guarantees that this linear program admits an integral optimum.

\subsection{Scheduling with Exponential Service Demands} 

We consider $m$ jobs, all available at time $0$, that can be processed over the unit time horizon $[0,1]$. Each job $i$ has a service demand (or job size) of $Z_i$ units of work, where the $Z_i$ are independent $\Exp(1)$ random variables. At each time $t$, a scheduler chooses a nonnegative service-rate vector $s(t)$, where $s_i(t)$ is the amount of work processed for job $i$ per unit of time. The feasible service-rate vectors form a known compact, convex, downward-closed set $P\subseteq\R_+^m$ containing the origin. 

For example, a single machine with unit processing capacity, which can be shared among jobs, corresponds to $P=\{s\in\R_+^m:\sum_{i=1}^m s_i\leq1\}$. The region $P$ can also encode asymmetric processing capabilities and more complex resource constraints. A constraint of the form $\sum_{i=1}^m a_i s_i\leq1$, with $a_i>0$, models jobs requiring different amounts $a_i$ of machine time per unit of processed work. Several such constraints can represent multiple shared resources, and additional coordinate bounds can impose job-specific service-rate limits.

The cumulative work processed for job $i$ by time $t$ is $x_i(t)\coloneqq\int_0^t s_i(u)\,du$. Job $i$ completes when $x_i(t)$ reaches $Z_i$, after which its service rate is zero. The scheduler observes the service provided and the identities and times of completed jobs, but not the service demands of unfinished jobs. A policy is nonanticipating if its decisions depend only on this observed history and independent internal randomization. \medskip

The objective is to maximize the expected total work processed by the deadline, $\E[\sum_{i=1}^m x_i(1)]$, including work performed on jobs that remain unfinished. This corresponds to the \emph{early-work} objective in scheduling which has appeared for instance in Sterna and Czerniachowska~\cite{SternaCzerniachowska2017} and Jiang et al.~\cite{JiangEtAl2024}. 

An online scheduling policy chooses a service-rate vector $s(t)\in P$ at each time $t$ and stops serving each job when it completes. Its choices may depend on the service and completion history observed so far and internal randomization independent of the service demands. Typically, scheduling problems focus on the expected number of jobs completed by the deadline. Both of those objectives coincide for online algorithms.  Related representations through hidden random thresholds arise in matching with stochastic rewards~\cite{MehtaPanigrahi2012,HuangZhang2024}. 

\begin{lemma}
\label{lem:work-completions}
For any online scheduling policy, let $x_i(1)$ denote the total work processed for job $i$ by time $1$. If the service demands are independent with $Z_i\sim\Exp(1)$, then
\begin{equation*}
\E\left[\sum_{i=1}^m x_i(1)\right]
=
\E\left[\sum_{i=1}^m \ind\{x_i(1)=Z_i\}\right].
\end{equation*}
\end{lemma}

\begin{proof}
Fix $i$ and set $Z^{(i)}\coloneqq(Z_1,\ldots,Z_{i-1},\infty,Z_{i+1},\ldots,Z_m)$. Run the same policy on $Z$ and $Z^{(i)}$ with the same internal randomization, and let $B_i$ be the work given to job $i$ by time $1$ under $Z^{(i)}$. Since $B_i$ depends only on the other demands and the randomization, it is independent of $Z_i$. If $Z_i>B_i$, both runs coincide throughout and job $i$ receives $B_i$ units of work. Otherwise, they coincide until it receives $Z_i$ units, when it completes in the actual run and service stops. Thus $x_i(1)=\min\{Z_i,B_i\}$. Therefore,
\begin{equation*}
\E[\min\{Z_i,B_i\}\mid B_i]=\int_0^\infty \Pp(\min\{Z_i,B_i\}>u\mid B_i)\,du=\int_0^{B_i}e^{-u}\,du
=
1-e^{-B_i}
=
\Pp(x_i(1)=Z_i\mid B_i).
\end{equation*}
Taking expectations and summing over $i$ proves the claim.
\end{proof}

The offline benchmark knows the entire vector of service demands in advance. For a realization $z\in\R_+^m$, its value is
\begin{equation*}
\OPTE_P(z)\coloneqq
\max_{x\in \R_+^m}\left\{
\sum_{i \in [m]} x_i:\ 0\leq x\leq z, \quad x\in P
\right\}.
\end{equation*}
Here $x_i$ represents a total allocation of work and the inequality is  understood coordinate-wise, whereas $s_i(t)$ is an instantaneous service rate. Both vectors are constrained by the same region $P$ because the horizon has length one: by convexity, the cumulative allocation $x \coloneqq x(1)=\int_0^1 s(t)\,dt$ belongs to $P$. Conversely, any feasible allocation $x$ in this program can be implemented by the constant service-rate vector $s(t)=x$. Over a horizon of length $T>0$, the corresponding constraint on total allocations would instead be $x\in T \cdot P$.

\begin{remark}
Processed work and job completions can differ on a given realization, even under online policies. For example, consider two jobs sharing a machine of unit capacity, with realized demands $Z_1=Z_2=3/4$. Serving the first job at rate one until completion and then serving the second processes one unit of work and completes one job. Serving both jobs at rate $1/2$ throughout the horizon also processes one unit of work, but completes neither job.

For independent unit-exponential demands, \Cref{lem:work-completions} shows that the expected total work processed equals the expected number of completed jobs under every online policy. This equality need not hold for clairvoyant policies, and an offline policy maximizing work need not maximize job completions. For independent heterogeneous demands $Z_i\sim\Exp(\mu_i)$, every online policy satisfies a weighted version of the equality in expectation, $\Pp(x_i(1)=Z_i)=\mu_i\E[x_i(1)]$ for each $i$. 
\end{remark}


\subsection{Equivalence of Online Policies through Flow Decomposition}

We now show an equivalence between a supply-matching allocation and a total service allocation under the same polytope. 
To connect the two models, define
\begin{equation*}
\calP\coloneqq
\left\{
x\in\R_+^m:\ x(A)\leq\Lambda(A)
\quad\forall A\subseteq[m]
\right\}.
\end{equation*}
The function $\Lambda$ is a normalized, monotone, submodular coverage function, so $\calP$ is its associated polymatroid. In particular, $\calP$ is compact, convex, downward closed, and contains the origin; it is a valid service rate set constraint. Operationally, the scheduler has access to servers indexed by $S\subseteq[m]$. Each server $S$ can divide a total processing rate of at most $\lambda_S$ among the jobs in $S$.
We specialize the scheduling model to this region and write $\OPTE(z)\coloneqq\OPTE_{\calP}(z)$.  Taking a feasible matching allocation (without the unit capacity constraint) and considering the aggregate for a given offline node/job clearly shows that the polymatroid constraints are necessary. The following flow decomposition lemma shows that they are in fact sufficient.

\begin{lemma}
\label{lem:flow-decomposition}
A vector $x\in\R_+^m$ belongs to $\calP$ if and only if there exist nonnegative allocations $(x_{i,S})_{i \in [m], S \subset [m]}$, such that
\begin{equation*}
\sum_{i\in S}x_{i,S}\leq\lambda_S
\quad\forall S\subseteq[m],
\qquad
\sum_{S\ni i}x_{i,S}=x_i
\quad\forall i\in[m].
\end{equation*}
\end{lemma}

\begin{proof}
Let $(x_{i,S})$ be a nonnegative allocation satisfying $\sum_{i\in S}x_{i,S}\leq\lambda_S$ for every $S\subseteq[m]$, and define $x_i\coloneqq\sum_{S\ni i}x_{i,S}$. For every $A\subseteq[m]$,
\begin{equation*}
x(A)
=
\sum_{S:\,S\cap A\neq\varnothing}
\sum_{i\in S\cap A}x_{i,S}
\leq
\sum_{S:\,S\cap A\neq\varnothing}\lambda_S
=
\Lambda(A).
\end{equation*}
So $x\in\calP$.

Conversely, fix $x\in\calP$. Consider a network with one source, one node for each vertex $i\in[m]$, one node for each type $S\subseteq[m]$, and one sink. Give the edge $\mathrm{source}\to i$ capacity $x_i$ and the edge $S\to\mathrm{sink}$ capacity $\lambda_S$. Add an edge $i\to S$ of infinite capacity when $i\in S$. It suffices to find a flow saturating all source edges. For such a flow, let $x_{i,S}$ be the flow on the edge $i\to S$. Flow conservation gives $\sum_{S\ni i}x_{i,S}=x_i$, while the capacities of the edges to the sink ensure $\sum_{i\in S}x_{i,S}\leq\lambda_S$.

A source--sink cut partitions the network nodes into two sides, with the source on one side and the sink on the other. Its capacity is the sum of the capacities of edges directed from the source side to the sink side. Let $A\subseteq[m]$ be the set of vertex nodes on the source side. For the cut to have finite capacity, every type node $S$ with $S\cap A\neq\varnothing$ must also lie on the source side. Otherwise, an infinite-capacity edge $i\to S$ with $i\in S\cap A$ would cross the cut. 

For this fixed $A$, a cheapest cut places all remaining type nodes on the sink side. Moving any of these nodes to the source side would add its nonnegative capacity $\lambda_S$ without removing any crossing edge. The edges crossing the resulting cut are therefore the source edges to vertices outside $A$ and the sink edges from types intersecting $A$. Since $x\in\calP$, their total capacity satisfies
\begin{equation*}
x([m]\setminus A)+\Lambda(A)
\geq
x([m]\setminus A)+x(A)
=
x([m]).
\end{equation*}

Every cut has capacity at least $x([m])$. The cut whose source side consists only of the source has capacity exactly $x([m])$, so this is the minimum cut capacity. By the max-flow/min-cut theorem, the maximum flow value is therefore $x([m])$. Since this equals the total capacity of the source edges, a maximum flow saturates all of them, which gives a valid $(x_{i,S})$ allocation.
\end{proof}

Thus $\calP$ describes the allocations supported by the mean type supplies, before imposing the unit capacities of the offline vertices. It also describes feasible instantaneous service rates when each type $S$ provides processing capacity $\lambda_S$ that can be shared among the jobs in $S$. Under the unit-horizon normalization, the same region $\calP$ describes total allocations from the mean type supplies in matching and instantaneous service rates in scheduling. 


The flow decomposition also gives a correspondence between online policies in the two models. Service rates can be implemented through the Poisson arrival streams, while matching decisions can conversely be represented by service rates.

\begin{lemma}
\label{lem:online-equivalence}
For Poisson arrivals $\boldsymbol{\lambda}$ and a service rate constraint $\calP$, every online policy in either model has a counterpart in the other such that the processes of matched vertices and completed jobs have the same law. For corresponding policies, the expected matching size equals the expected total work processed.
\end{lemma}

\begin{proof}
First, fix an online scheduling policy with feasible-rate region $\calP$. By the flow decomposition lemma, choose nonnegative rates $s_{i,S}(t)$\footnote{The decomposition can be chosen as a measurable function of $s(t)$.} satisfying
\begin{equation*}
\sum_{i\in S}s_{i,S}(t)\leq\lambda_S
\quad\forall S\subseteq[m],
\qquad
\sum_{S\ni i}s_{i,S}(t)=s_i(t)
\quad\forall i\in[m].
\end{equation*}

The matching policy runs the scheduling rule, tracking its cumulative service and interpreting previous matches, with their vertex identities and times, as job completions. When a type-$S$ vertex arrives, with $\lambda_S>0$, match it to $i\in S$ with probability $s_{i,S}(t)/\lambda_S$, using the rates immediately before the arrival, and reject it with the remaining probability. The decomposition ensures that these probabilities sum to at most one. Since completed jobs receive no service, matched vertices receive no further assignments.

The conditional matching rate of an unmatched vertex $i$ is $\sum_{S\ni i}s_{i,S}(t)=s_i(t)$. By exponential memorylessness, this is also the conditional completion rate of the corresponding unfinished job. Use the same randomization for the scheduling rule in both constructions. From any shared history, their rates agree until the next event, so its time and identity have the same conditional law. Starting from the empty history, we can therefore couple successive matches and completions. Thus the vertex matching times have the same joint distribution as the corresponding job completion times, assigning time $\infty$ whenever a vertex remains unmatched or a job remains unfinished at the deadline.

Conversely, fix an online matching policy. Let $p_{i,S}(t)$ be its probability of assigning a type-$S$ arrival at time $t$ to $i$, conditional only on the identities and times of previous matches. These probabilities can be computed from the given policy and the known arrival law. Using completion history in place of matching history, define the scheduling rule $s_i(t)=\sum_{S\ni i}\lambda_Sp_{i,S}(t)$. These rates belong to $\calP$ by the flow decomposition lemma and vanish for completed jobs. The same coupling argument gives equality of the joint laws of matching and completion times.

Finally, corresponding policies have the same expected number of matches and job completions. Lemma~\ref{lem:work-completions} therefore gives
\begin{equation*}
\E[\text{number of matches by time }1]
=
\E\left[\sum_{i=1}^m x_i(1)\right]. \qedhere
\end{equation*}
\end{proof}
Given the above Lemma, for a fixed pair of corresponding online policies, we write $\ALG$ for their common expected reward.

For $P=\calP$, the two models therefore have the same optimal expected online reward. To compare their competitive ratios, it remains to relate their offline optima. In the next section, we prove that $\E[\OPTP(\mathbf N)]=\E[\OPTE(Z)]$. Together with the online correspondence, this identity allows competitive ratio guarantees to be transferred between corresponding matching and scheduling instances.

\subsection{A Shorter Proof for Exchangeable Arrivals}
\label{subsec:exchangeable-warmup}

Before treating general instances, we show how to heuristically derive the tight competitive ratio in the exchangeable case. Here a scalar ODE describes the online process, and an asymptotic matching formula gives the offline benchmark. This calculation also motivates the comparison ODE used in our general analysis. \medskip

The constant $1-e^{1-e}$ has a direct interpretation in exchangeable instances, namely those for which $\lambda_S=\lambda_T$ whenever $|S|=|T|$. Suppose arrivals have total rate $m$, and each independently samples $D$ uniform neighbors with replacement, where $D$ has a fixed generating function $\phi$. Greedy is optimal among online policies in this model. When a fraction $x$ of the offline vertices is matched, an arrival with $D$ sampled neighbors cannot be matched with probability $x^D$. The limiting matched fraction therefore satisfies $x'(t)=1-\phi(x(t))$, with $x(0)=0$.

The bipartite matching limit of \cite[Theorem~3]{BordenaveLelargeSalez2013}, specialized to this model, gives
\begin{equation*}
\lim_{m\to\infty}\frac{\E[\OPTP(\mathbf N)]}{m}
=
1-\max_{u\in[0,1]}\{\phi(u)-\Psi(u)\},
\end{equation*}
where $\Psi(u)\coloneqq u+(1-u)\log(1-u)$, with $\Psi(1)=1$.\footnote{We give the specialization in \Cref{app:exchangeable-limit}.}
The expansion $\Psi(u)=\sum_{d=2}^{\infty}u^d/(d(d-1))$ shows that $\Psi$ is itself a generating function, corresponding to the hard instance of \cite{ChierichettiEtAl2025}. The expected offline optimum is $m-o(m)$ precisely when $\phi(u)\leq\Psi(u)$ for every $u\in[0,1]$.

To make greedy progress as slowly as possible while retaining this offline value, we therefore choose the largest admissible generating function, namely $\phi=\Psi$. The ODE becomes $x'=(1-x)(1-\log(1-x))$. Setting $h(t)=-\log(1-x(t))$ gives $h'=1+h$, hence $x(t)=1-e^{1-e^t}$ and $x(1)=1-e^{1-e}$.

A short ODE comparison gives the same asymptotic guarantee for a general $\phi$. Let $\beta>0$ denote its scaled offline limit, and let $x$ now solve $x'=1-\phi(x)$, with $x(0)=0$. By definition of $\beta$, we have $\phi(u)\leq1-\beta+\Psi(u)$. Since $\Psi$ is convex and $\Psi(0)=0$, for every $u\in[0,1]$,
\begin{equation*}
1-\phi(\beta u)
\geq
\beta-\Psi(\beta u)
\geq
\beta\bigl(1-\Psi(u)\bigr).
\end{equation*}
For $y(t)\coloneqq\beta(1-e^{1-e^t})$, the function $y/\beta$ solves the ODE associated with $\Psi$, so the previous inequality gives $y'(t)=\beta(1-\Psi(y(t)/\beta))\leq1-\phi(y(t))$. Since $x'=1-\phi(x)$ and both trajectories start at zero, comparison gives $x(t)\geq y(t)$. At $t=1$, this gives the asymptotic ratio $x(1)/\beta\geq1-e^{1-e}$, with equality when $\phi=\Psi$.

In the rest of the paper, we allow arbitrary arrival laws without assuming exchangeability.

\section{A Combinatorial Poisson--Exponential Duality}
\label{sec:static-duality}

The online correspondence established in the previous section preserves expected rewards. We now show that the expected offline optima also coincide.

\begin{theorem}
\label{thm:offline-duality}
For Poisson matching with intensity vector $\boldsymbol\lambda$ and scheduling with feasible-rate region $\calP$ and $\Exp(1)$ service demands,
\begin{equation*}
\E[\OPTP(\mathbf N)]
=
\E[\OPTE(Z)].
\end{equation*}
\end{theorem}

This identity is immediate when there is a single offline vertex. In that case, the matching optimum is $\min\{N,1\}$, where $N\sim\Pois(\lambda)$, while the scheduling optimum is $\min\{\lambda,Z\}$, where $Z\sim\Exp(1)$. Thus
\begin{equation*}
\E[\min\{N,1\}]=
\Pp(N\geq1)
=
1-e^{-\lambda}
=
\int_0^\lambda e^{-u}\,du
=\int_0^\infty \Pp(\min\{\lambda,Z\}>u)\,du=
\E[\min\{\lambda,Z\}].
\end{equation*}
Therefore, what the theorem establishes is that this simple identity extends to the combinatorial settings at hand.

We first express both offline optima through a max-deficiency characterization. We then prove the general result that the law of the maximum of a supermodular function is characterized by its law on chains, i.e. nested sets. On each such chain, the duality between Poisson counts and exponential waiting times for Poisson point processes gives an exact comparison, which we combine with the previous supermodular characterization result to prove the theorem.

\subsection{Offline Optima and Deficiencies}
\label{subsec:two-deficiencies}

For a set $A$ of offline vertices, only the $N(A)$ arrivals with a neighbor in $A$ can be matched to its vertices. Hence at least $|A|-N(A)$ vertices must remain unmatched whenever this difference is positive. Hall's theorem states that all offline vertices can be matched if and only if $N(A)\geq |A|$ for every $A\subseteq[m]$ \cite[Theorem~22.1]{Schrijver2003}. Ore's deficiency formula, which we re-derive below for self-containedness, gives the quantitative version of this result, identifying the maximum of these deficits with the minimum number of unmatched vertices \cite[Theorem~22.2]{Schrijver2003}.

The scheduling problem has a similar deficit characterization. The jobs in $A$ have total demand $Z(A)$, but the region $\calP$ allows at most $\Lambda(A)$ units of work to be allocated to them over the unit horizon. We therefore define
\begin{equation*}
\dP(A)\coloneqq |A|-N(A),
\qquad
\dE(A)\coloneqq Z(A)-\Lambda(A).
\end{equation*}
Both deficiencies are equal to $0$ at $A=\emptyset$, so their maxima over subsets of $[m]$ are nonnegative. The following lemma gives the two offline optima in terms of these deficiencies, which we obtain continuing the flow network analogy from earlier.

\begin{lemma}
\label{lem:two-deficiencies}
For every realization of the type counts and service demands,
\begin{equation*}
\OPTP(\mathbf N)=
m-\max_{A\subseteq[m]}\dP(A),\qquad
\OPTE(Z)=
Z([m])-\max_{A\subseteq[m]}\dE(A).
\end{equation*}
\end{lemma}

\begin{proof}
Consider the network from \Cref{lem:flow-decomposition}, with capacity one on each source edge to an offline vertex and capacity $N_S$ on the edge from type $S$ to the sink. The edge $i\to S$ has infinite capacity whenever $i\in S$. An integral flow corresponds to a matching, with $x_{i,S}$ equal to the flow from vertex $i$ to type $S$.

Let $A$ be the set of offline vertices on the source side of a cut. For the cut to have finite capacity, every type intersecting $A$ must also lie on the source side. A cheapest cut for this fixed $A$ puts all other type nodes on the sink side. Its capacity is $m-|A|+N(A)$. By max-flow/min-cut and integrality,
\begin{equation*}
\OPTP(\mathbf N)
=
\min_{A\subseteq[m]}\{m-|A|+N(A)\}
=
m-\max_{A\subseteq[m]}\dP(A).
\end{equation*}

For scheduling, replace the source capacity of vertex $i$ by $Z_i$ and the sink capacity of type $S$ by $\lambda_S$. By \Cref{lem:flow-decomposition}, feasible flows correspond exactly to allocations $x\in\calP$ satisfying $x\leq Z$, so the maximum flow value is $\OPTE(Z)$. The cheapest cut with vertex set $A$ on the source side now has capacity $Z([m]\setminus A)+\Lambda(A)$. Consequently,
\begin{equation*}
\OPTE(Z)
=
\min_{A\subseteq[m]}
\{Z([m]\setminus A)+\Lambda(A)\}
=
Z([m])-\max_{A\subseteq[m]}\dE(A). \qedhere
\end{equation*}
\end{proof}

Since $\E[Z([m])]=m$, equality of the expected offline optima reduces to equality of the expected maximum deficiencies. The functions $N$ and $\Lambda$ are submodular coverage functions, whereas $A\mapsto |A|$ and $A\mapsto Z(A)$ are modular. Both deficiencies are therefore supermodular.

\subsection{Supermodular Maxima through Chains}
\label{subsec:local-global}

Recall that a set function $d$ is supermodular if $d(A)+d(B)\leq d(A\cup B)+d(A\cap B)$ for every pair of sets $A,B$. The next proposition shows that the law of its maximum is determined by joint exceedance probabilities along chains of nested sets.

\begin{proposition}
\label{lem:supermodular-chains}
Let $d:2^{[m]}\to\R$ be a random supermodular function with $d(\varnothing)=0$ almost surely. For every $a\geq0$,
\begin{equation*}
\Pp\left(\max_{A\subseteq[m]}d(A)>a\right)
=
\sum_{k=1}^{m}(-1)^{k-1}
\sum_{\varnothing\neq A_1\subsetneq\cdots\subsetneq A_k\subseteq[m]}
\Pp\left(d(A_r)>a\ \text{for all }r\right).
\end{equation*}
\end{proposition}

Before giving the proof, we recall a few rudimentary notions about simplicial complexes. Although the proof can be carried out using elementary computations, this geometric viewpoint helps explain why the characterization has such a form. A finite simplicial complex is a finite collection of vertex sets, called simplices, that is closed under taking subsets. In other words, it is a downward-closed combinatorial set system. A simplex with $k$ vertices has dimension $k-1$. The geometric realization of the complex represents its simplices as points, edges, filled triangles, and higher-dimensional simplices, glued together along their common faces. For example, write $ab$ for $\{a,b\}$, and similarly for other vertex sets, omitting the empty simplex from the following lists. The collection $\{ab,bc,ac\}$ is not a simplicial complex, since it contains edges but omits their vertices. Adding these vertices gives the complex $\{a,b,c,ab,bc,ac\}$, whose geometric realization is the boundary of a triangle. Adding the simplex $abc$ gives $\{a,b,c,ab,bc,ac,abc\}$, whose realization is the filled triangle. \medskip

The Euler characteristic of a finite complex $K$ is the alternating sum
\begin{equation*}
\chi(K)\coloneqq\sum_{\varnothing\neq\sigma\in K}(-1)^{|\sigma|-1}.
\end{equation*}
The empty complex has Euler characteristic zero. For a complex consisting of vertices, edges, and filled triangles, this gives $\chi=V-E+T$, where $V$, $E$, and $T$ count these simplices. Thus the boundary of a triangle has Euler characteristic zero, while the filled triangle has Euler characteristic one.

The Euler characteristic is useful in that it connects these combinatorial counts to the topology of the geometric realization. A familiar example is Euler's formula $V-E+F=2$ for connected planar graphs. Here $F$ counts all faces, including the outer face, and the formula gives the Euler characteristic of the sphere decomposed into these vertices, edges, and faces. \medskip 

For our proof, the relevant topological property is contractibility. A complex is contractible if its geometric realization can be continuously deformed to a point within itself. Every nonempty contractible finite simplicial complex has Euler characteristic one \cite{Bjorner1995}.

The order complex of a finite partially ordered set has one vertex for each element and one simplex for each strict chain. Its Euler characteristic is therefore an alternating sum over nonempty chains, with a chain of $k$ elements contributing $(-1)^{k-1}$.

We apply this construction to the family of sets $A\subseteq[m]$ satisfying $d(A)>a$, ordered by inclusion, where $d$ is a supermodular function and $a>0$ is a fixed threshold. It remains to show that this order complex is contractible whenever it is nonempty.

\begin{proof}
Fix a realization of $d$ and a threshold $a\geq0$. Let $Q_a\coloneqq\{A\subseteq[m]:d(A)>a\}$, ordered by inclusion, and let $\mathcal K_a$ be its order complex. Since $d(\varnothing)=0$, the empty set does not belong to $Q_a$. A simplex of $\mathcal K_a$ is thus a chain $\varnothing\neq A_1\subsetneq\cdots\subsetneq A_k\subseteq[m]$ whose sets all satisfy $d(A_r)>a$. It follows that
\begin{equation*}
\chi(\mathcal K_a)
=
\sum_{k=1}^{m}(-1)^{k-1}
\sum_{\varnothing\neq A_1\subsetneq\cdots\subsetneq A_k\subseteq[m]}
\ind\{d(A_r)>a\ \text{for all }r\}.
\end{equation*}
We will show that $\mathcal K_a$ is contractible whenever it is nonempty. \medskip 

Assume $Q_a\neq\varnothing$ and choose $B$ maximizing $d$ over all subsets of $[m]$. Then $B\in Q_a$. For every $A\in Q_a$, supermodularity and maximality of $B$ give
\begin{equation*}
d(A\cup B)
\geq d(A)+d(B)-d(A\cap B)
\geq d(A)>a.
\end{equation*}
Thus $A\cup B$ also belongs to $Q_a$.

We delete, one at a time, the vertices $A\in Q_a$ that do not contain $B$, in decreasing order of cardinality. Consider such a vertex $A$ when it is about to be deleted, and put $C\coloneqq A\cup B$. The set $C$ belongs to $Q_a$ and has not been deleted, since it contains $B$. Every remaining chain containing $A$ can also contain $C$. Indeed, any set below $A$ is contained in $C$. Any remaining strict superset $U$ of $A$ must contain $B$, because otherwise it would have been deleted earlier. Hence $C=A\cup B\subseteq U$. This shows that adding $C$ to the chain still gives a chain.

We can therefore delete $A$ by continuously moving its weight onto $C$. Let $e_D$ denote the coordinate vector associated with vertex $D$. In this geometric realization, every point has the form $p=\sum_{D\in\sigma}\alpha_De_D$ for some simplex $\sigma$, with nonnegative coefficients summing to one. Set $\alpha_A=0$ if $A\notin\sigma$, and define
\begin{equation*}
H_t(p)\coloneqq p+t\alpha_A(e_C-e_A),
\qquad 0\leq t\leq1.
\end{equation*}
If $A\notin\sigma$, the point remains fixed. Otherwise, the chain property above ensures that $\sigma\cup\{C\}$ is also a simplex. Both endpoints of the motion belong to its geometric realization. Indeed, $H_0(p)=p$ belongs to the face $\sigma$, while $H_1(p)
=
\sum_{D\in\sigma\setminus\{A\}}\alpha_De_D+\alpha_Ae_C$ is a convex combination of vertices of $\sigma\cup\{C\}$. Since $H_t(p)=(1-t)H_0(p)+tH_1(p)$ and this geometric simplex is convex, the entire motion stays inside it. At $t=1$, the coefficient of $e_A$ is zero, so the image lies in the subcomplex obtained by deleting $A$. The map is continuous in $(p,t)$ and fixes every point whose $A$-coordinate is already zero. Thus it defines a deformation retraction onto this subcomplex.

After all these deletions, every remaining vertex contains $B$, and $B$ itself remains. It can therefore be added to every remaining chain. The remaining complex is a cone with apex $B$ which therefore can be continuously contracted to $B$. Hence the original complex $\mathcal K_a$ is contractible. We have proved that $\chi(\mathcal K_a)$ equals one when $\max_A d(A)>a$ and zero otherwise. Taking expectations in the alternating chain sum proves the proposition.
\end{proof}

The topological argument is classical. The key observation here is that supermodularity ensures that the order complex of the violating sets is contractible whenever it is nonempty. We can thus compare the Poisson and exponential deficiencies through identities on individual chains, without requiring their joint laws to agree over all subsets. We establish the required chain identity next.

\subsection{Poisson--Exponential Duality on a Chain}
\label{subsec:chain-duality}

We now compare the deficiencies along a fixed chain. The argument uses the relation between Poisson counts and exponential waiting times. If $\Pi$ is a Poisson process of rate one, with $\Pi(t)$ counting the number of arrivals up to time $t$, and $T_n$ is its $n$-th arrival time, then $T_n>t$ if and only if $\Pi(t)<n$. Moreover, the successive interarrival times are independent unit-exponential random variables.

\begin{lemma}
\label{lem:chain-duality}
Fix a chain $\varnothing\neq A_1\subsetneq\cdots\subsetneq A_k\subseteq[m]$. For every $a\geq0$, let $K_a\sim\Pois(a)$ be independent of $\mathbf N$. Then
\begin{equation*}
\Pp\left(\dE(A_r)>a\ \text{for all }r\right)
=
\Pp\left(\dP(A_r)>K_a\ \text{for all }r\right).
\end{equation*}
\end{lemma}

\begin{proof}
Set $A_0=\varnothing$. Along the chain, we view $t_r\coloneqq\Lambda(A_r)$ as the observation times of an auxiliary rate-one Poisson process. Since $A_{r-1}\subseteq A_r$, every type intersecting $A_{r-1}$ also intersects $A_r$. Splitting the types according to whether they already intersect $A_{r-1}$ gives
\begin{equation*}
N(A_r)=\sum_{S:\,S\cap A_r\neq\varnothing}N_S=\sum_{S:\,S\cap A_{r-1}\neq\varnothing}N_S
+\sum_{\substack{S:\,S\cap A_r\neq\varnothing\\
                     S\cap A_{r-1}=\varnothing}}N_S=N(A_{r-1})
+\sum_{\substack{S:\,S\cap A_r\neq\varnothing\\
                     S\cap A_{r-1}=\varnothing}}N_S.
\end{equation*}
Thus each increment $N(A_r)-N(A_{r-1})$ is a sum of independent Poisson random variables and is therefore Poisson. Applying the same decomposition to $\Lambda$ shows that its mean is $\Lambda(A_r)-\Lambda(A_{r-1})=t_r-t_{r-1}$. Moreover, each type contributes only when its neighborhood first intersects the chain, so different increments involve disjoint collections of type counts and are independent.

Now recall that, for a rate-one Poisson process $\Pi$, the number of points arriving between times $t$ and $t'$ has distribution $\Pi(t')-\Pi(t)\sim\Pois(t'-t)$, and increments over disjoint time intervals are independent. Since $0=\Lambda(A_0)\leq\Lambda(A_1)\leq\cdots\leq\Lambda(A_k)$, the increments $\Pi(\Lambda(A_r))-\Pi(\Lambda(A_{r-1}))$ are independent and have exactly the same distributions as the increments above. Both processes start at zero, so summing successive increments gives the joint identity
\begin{equation*}
\bigl(N(A_r)\bigr)_{r=1}^k
\stackrel{d}{=}
\bigl(\Pi(\Lambda(A_r))\bigr)_{r=1}^k.
\end{equation*}

The exponential demands have a similar representation through the arrival times of $\Pi$, with $|A_r|$ now playing the role of an arrival index. Let $T_n$ denote the $n$-th arrival time, with $T_0=0$. The interarrival times $T_j-T_{j-1}$ are independent unit-exponential random variables.
 We again compare increments along the chain. Since $A_{r-1}\subseteq A_r$,
\begin{align*}
Z(A_r)=\sum_{i\in A_r}Z_i=\sum_{i\in A_{r-1}}Z_i+\sum_{i\in A_r\setminus A_{r-1}}Z_i=Z(A_{r-1})+\sum_{i\in A_r\setminus A_{r-1}}Z_i.
\end{align*}
Thus $Z(A_r)-Z(A_{r-1})$ is a sum of $|A_r|-|A_{r-1}|$ independent unit-exponential variables. Moreover, the sets $A_r\setminus A_{r-1}$ are disjoint, so these increments are independent across $r$.

On the Poisson-process side, the time between arrivals $|A_{r-1}|$ and $|A_r|$ is
\begin{equation*}
T_{|A_r|}-T_{|A_{r-1}|}
=
\sum_{j=|A_{r-1}|+1}^{|A_r|}(T_j-T_{j-1}).
\end{equation*}
This is also a sum of $|A_r|-|A_{r-1}|$ independent unit-exponential variables. Different values of $r$ use disjoint blocks of interarrival times, so these increments are independent as well. The two sequences therefore have the same joint increment law. Since $Z(A_0)=T_0=0$, summing successive increments gives
\begin{equation*}
\bigl(Z(A_r)\bigr)_{r=1}^k
\stackrel{d}{=}
\bigl(T_{|A_r|}\bigr)_{r=1}^k.
\end{equation*}

We now use the elementary equivalence $T_n>t$ if and only if $\Pi(t)<n$. This gives
\begin{align*}
\Pp\bigl(\dE(A_r)>a\ \text{for all }r\bigr)
&=\Pp\bigl(Z(A_r)>\Lambda(A_r)+a\ \text{for all }r\bigr)\\
&=\Pp\bigl(T_{|A_r|}>\Lambda(A_r)+a\ \text{for all }r\bigr)\\
&=\Pp\bigl(\Pi(\Lambda(A_r)+a)<|A_r|\ \text{for all }r\bigr).
\end{align*}

To remove the common time shift $a$, write $\Pi(\Lambda(A_r)+a)=\Pi(a)+[\Pi(a+\Lambda(A_r))-\Pi(a)]$. By stationary independent increments, the bracketed terms have the same joint law as $(\Pi(\Lambda(A_r)))_{r=1}^k$ and are independent of the initial count $\Pi(a)\sim\Pois(a)$. Thus, taking $K_a\sim\Pois(a)$ independent of both $\Pi$ and $\mathbf N$, and using the Poisson representation above, gives
\begin{align*}
\Pp\bigl(\Pi(\Lambda(A_r)+a)<|A_r|\ \text{for all }r\bigr)
&=\Pp\bigl(K_a+\Pi(\Lambda(A_r))<|A_r|\ \text{for all }r\bigr)\\
&=\Pp\bigl(K_a+N(A_r)<|A_r|\ \text{for all }r\bigr)\\
&=\Pp\bigl(\dP(A_r)>K_a\ \text{for all }r\bigr).
\end{align*}
The same variable $K_a$ appears in every inequality because the initial interval $[0,a]$ is shared by all observation times.
\end{proof}

The chain condition is essential to this argument. Nested sets give independent increments in both representations. 
In fact, the identity does not hold for incomparable sets. For example, take $m=2$ with a single arrival type $\{1,2\}$ of intensity one. Then $N(\{1\})=N(\{2\})=N(\{1,2\})$, while $Z_1$ and $Z_2$ are independent. Hence, for every $a\geq0$,
\begin{equation*}
\Pp\bigl(\dE(\{1\})>a,\ \dE(\{2\})>a\bigr)
=e^{-2(1+a)},\qquad 
\Pp\bigl(\dP(\{1\})>K_a,\ \dP(\{2\})>K_a\bigr)
=e^{-(1+a)}.
\end{equation*}
Indeed, the first event requires $Z_1>\Lambda(\{1\})+a=1+a$ and $Z_2>\Lambda(\{2\})+a=1+a$, whereas the second is $\vert \{1\}\vert -N(\{1\})>K_a$ and $\vert \{2\}\vert -N(\{2\})>K_a$ which is equivalent to $N_{\{1,2\}}=K_a=0$.

\subsection{From Chain Duality to Offline Equivalence}
\label{subsec:global-duality}

Combining the two preceding results gives a distributional relation between the maximum deficiencies.

\begin{proposition}
\label{thm:deficiency-duality}
For every $a\geq0$, let $K_a\sim\Pois(a)$ be independent of $\mathbf N$. Then
\begin{equation}
\Pp\bigg(\max_{A\subseteq[m]}\dE(A)>a\bigg)
=
\Pp\bigg(\max_{A\subseteq[m]}\dP(A)>K_a\bigg).
\label{eq:tail-duality}
\end{equation}
In particular, for independent $E_1,E_2,\ldots\sim\Exp(1)$, taking the empty sum to be zero,
\begin{equation}
\max_{A\subseteq[m]}\dE(A)
\stackrel{d}{=}
\sum_{\ell=1}^{\max\limits_{A\subseteq[m]}\dP(A)} \!\!\!\!\! \!\!\!\!E_\ell \sim \operatorname{Gamma}\bigg(\max_{A\subseteq[m]}\dP(A),1\bigg).
\label{eq:gamma-mixture}
\end{equation}
\end{proposition}

\begin{proof}
Fix $a\geq0$ and let $K_a\sim\Pois(a)$ be independent of $\mathbf N$. All chain sums below range over nonempty subsets of $[m]$. Applying \Cref{lem:supermodular-chains} to $\dE$, then replacing each chain probability using \Cref{lem:chain-duality}, gives
\begin{align*}
\Pp\left(\max_{A\subseteq[m]}\dE(A)>a\right)
&=
\sum_{k=1}^m(-1)^{k-1}
\sum_{A_1\subsetneq\cdots\subsetneq A_k}
\Pp\bigl(\dE(A_r)>a\ \text{for all }r\bigr)\\
&=
\sum_{k=1}^m(-1)^{k-1}
\sum_{A_1\subsetneq\cdots\subsetneq A_k}
\Pp\bigl(\dP(A_r)>K_a\ \text{for all }r\bigr)\\
&=
\E\left[
\sum_{k=1}^m(-1)^{k-1}
\sum_{A_1\subsetneq\cdots\subsetneq A_k}
\Pp\bigl(\dP(A_r)>K_a\ \text{for all }r\mid K_a\bigr)
\right]\\
&=
\E\left[
\Pp\left(\max_{A\subseteq[m]}\dP(A)>K_a\mid K_a\right)
\right]\\
&=
\Pp\left(\max_{A\subseteq[m]}\dP(A)>K_a\right).
\end{align*}
To apply \Cref{lem:supermodular-chains} at the random threshold $K_a$, we condition on $K_a$. Its independence from $\mathbf N$ leaves the law of $\dP$ unchanged, so we can apply the proposition with $K_a$ held fixed. This proves \eqref{eq:tail-duality}.

For the second statement, condition on $\mathbf N$. This fixes the number of terms in the sum without changing the joint law of the independent unit-exponential variables $E_\ell$. Using again the Poisson point process equivalence $T_n>a$ if and only if $\Pi(a)<n$, and then averaging over $\mathbf N$, gives
\begin{align*}
\Pp(\!\!\!\!\!\!\sum_{\ell=1}^{\max\limits_{A\subseteq[m]}\dP(A)} \!  \!\!\!\!\!\! \!\!\!E_\ell>a)
=
\E[
\Pp(K_a<\max_{A\subseteq[m]}\dP(A)\mid\mathbf N)
]
=
\Pp(\max_{A\subseteq[m]}\dP(A)>K_a)=
\Pp(\max_{A\subseteq[m]}\dE(A)>a),
\end{align*}
where the last equality follows from \eqref{eq:tail-duality}. The first equality also covers the empty sum, since both the event that this sum exceeds $a\geq0$ and the event $K_a<0$ are impossible. Both variables are nonnegative, so equality of their tail probabilities for every $a\geq0$ proves \eqref{eq:gamma-mixture}.
\end{proof}

The Poisson maximum deficiency is the number of vertices left unmatched by an optimal matching. The exponential maximum deficiency is the amount of demand left unprocessed by an optimal offline allocation. Thus \eqref{eq:gamma-mixture} represents the latter as a sum of unit-exponential quantities indexed by the former. We can now prove the main Theorem regarding offline equality. 

\begin{proof}[Proof of \Cref{thm:offline-duality}]
In \eqref{eq:gamma-mixture}, each exponential variable has mean one and is independent of the number of summands. Hence the two maximum deficiencies have the same expectation. Since $\E[Z([m])]=m$, \Cref{lem:two-deficiencies} gives
\begin{equation*}
\E[\OPTE(Z)]=
m-\E\left[\max_{A\subseteq[m]}\dE(A)\right]=
m-\E\left[\max_{A\subseteq[m]}\dP(A)\right]
=
\E[\OPTP(\mathbf N)]. \qedhere
\end{equation*}
\end{proof}

We write $\OPT\coloneqq\E[\OPTP(\mathbf N)]=\E[\OPTE(Z)]$ for the common expected offline optimum. Together with \Cref{lem:online-equivalence}, this shows that a guarantee $\ALG\geq\alpha\OPT$ transfers between corresponding policies in the two models.

Setting $a=0$ in \eqref{eq:tail-duality} also gives $\Pp(\OPTP(\mathbf N)=m)=\Pp(Z\in\calP)$. Thus the probability that all offline vertices can be matched equals the probability that all service demands can be processed by the deadline.

\subsection{Weighted Offline Equality}
\label{subsec:weighted-static}

The equality of expected offline optima extends to nonnegative edge rewards. For each type $S\subseteq[m]$ and vertex $i\in S$, let $w_{i,S}\geq0$ be the reward for matching a type-$S$ arrival to $i$. Write $\OPTP_w(\mathbf N)$ for the maximum total matching reward under the realized counts $\mathbf N$.

On the scheduling side, each type $S$ corresponds to a server that can divide a total processing rate of at most $\lambda_S$ among the jobs in $S$. Each unit of work it gives to job $i$ earns reward $w_{i,S}$. For a demand vector $z\in\R_+^m$, the offline allocation optimum is
\begin{equation*}
\begin{aligned}
\OPTE_w(z)\coloneqq
\max_{\substack{x_{i,S}\geq0\\
                 S\subseteq[m],\,i\in S}}
\quad&
\sum_{S\subseteq[m]}\sum_{i\in S}w_{i,S}x_{i,S}\\
\text{subject to}\quad&
\sum_{i\in S}x_{i,S}\leq\lambda_S
\qquad\forall S\subseteq[m],\\
&
\sum_{S\ni i}x_{i,S}\leq z_i
\qquad\forall i\in[m].
\end{aligned}
\end{equation*}
Here $x_{i,S}$ is the total work allocated by server $S$ to job $i$ over the unit horizon. The weighted matching optimum has the same formulation with type capacities $N_S$ and unit vertex capacities.

\begin{theorem}
\label{thm:weighted-static}
For every nonnegative edge-reward vector $w$ and independent service demands $Z_i\sim\Exp(1)$,
\begin{equation*}
\E[\OPTP_w(\mathbf N)]
=
\E[\OPTE_w(Z)].
\end{equation*}
This also holds for types with the same neighborhood but different edge weights.
\end{theorem}

The proof in \Cref{app:weighted-offline} expresses the optima through dual prices on offline vertices and builds a simplicial complex from chains of price vectors. An Euler-characteristic identity reduces the comparison of expected optima to local comparisons on this complex, which follow from the same Poisson count--waiting-time identity.

\paragraph{Random Edge Weights:} We can also handle random edge rewards with arbitrary dependence within each arrival, by treating each possible reward vector as a distinct type with the same neighborhood and deterministic edge rewards.
Suppose each type-$S$ arrival draws a nonnegative reward vector $W_S=(W_{i,S})_{i\in S}$ with law $F_S$, independently across arrivals and of the arrival processes. The vector is observed on arrival, and its coordinates may be arbitrarily dependent. Write $F=(F_S)_{S\subseteq[m]}$. Let $\mathcal A$ denote the realized collection of arrivals and their reward vectors, and write $\OPTP_F(\mathcal A)$ for the maximum total matching reward. If $F_S$ has finite support, each profile $w$ defines a subtype $(S,w)$ with neighborhood $S$, intensity $\lambda_S F_S(\{w\})$, and deterministic edge rewards $w_i$. These subtypes form independent Poisson streams. Thus specifying the refined type fixes both the eligible vertices and the rewards, while different refined types may still share the same neighborhood.

On the scheduling side, each refined type becomes a server with the corresponding deterministic processing rate and rewards. More generally, $F_S$ specifies the composition of a pool with total processing rate $\lambda_S$. The portion with reward profiles in a measurable set $B$ has processing rate $\lambda_S F_S(B)$. For continuous distributions, this gives a fluid model with continuously heterogeneous processing capacity. 
The pool composition determined by $F_S$ and the reward associated with each profile are known and deterministic, only the service demands are random.

At time $t$, the scheduler chooses measurable fractions $u_{i,S}(t,w)\geq0$, satisfying $\sum_{i\in S}u_{i,S}(t,w)\leq1$ for $F_S$-almost every $w$. These choices give
\begin{align*}
s_i(t)
&=
\sum_{S\ni i}\lambda_S
\int_{\R_+^S} u_{i,S}(t,w)\,F_S(dw),\\
\text{reward rate at time }t
&=
\sum_{S\subseteq[m]}\lambda_S
\int_{\R_+^S}\sum_{i\in S}w_i u_{i,S}(t,w)\,F_S(dw).
\end{align*}
Job $i$ receives no further service once its cumulative work reaches $Z_i$. 

For a realized demand vector $z$, the offline scheduler knows the demands and chooses the allocation fractions over the entire horizon. Its optimum is
\begin{equation}
\begin{aligned}
\OPTE_F(z)\coloneqq
\sup_u\quad&
\int_0^1
\sum_{S\subseteq[m]}\lambda_S
\int_{\R_+^S}\sum_{i\in S}w_i u_{i,S}(t,w)\,F_S(dw)\,dt\\
\text{subject to}\quad&
\sum_{i\in S}u_{i,S}(t,w)\leq1
\qquad\text{for a.e. }(t,w),\quad\forall S\subseteq[m],\\
&
\int_0^1\sum_{S\ni i}\lambda_S
\int_{\R_+^S} u_{i,S}(t,w)\,F_S(dw)\,dt
\leq z_i
\qquad\forall i\in[m].
\end{aligned}
\label{eq:random-reward-allocation}
\end{equation}
The supremum runs over measurable functions $u_{i,S}:[0,1]\times\R_+^S\to[0,1]$. The first constraint holds almost everywhere with respect to $dt\,F_S(dw)$ and limits the instantaneous use of each capacity profile. The second bounds the cumulative work given to each job. The objective is the total reward earned over the horizon.

\begin{corollary}
\label{cor:random-edge-rewards}
If $\sum_{S\subseteq[m]}\lambda_S\E[\|W_S\|_1]<\infty$, then, for independent $Z_i\sim\Exp(1)$,
\begin{equation*}
\E[\OPTP_F(\mathcal A)]
=
\E[\OPTE_F(Z)].
\end{equation*}
\end{corollary}

\begin{proof}
First suppose each $F_S$ has finite support. Apply \Cref{thm:weighted-static} to the refined types $(S,w)$ described above. Every feasible schedule induces subtype allocations
\begin{equation*}
x_{i,(S,w)}
=
\lambda_S F_S(\{w\})
\int_0^1u_{i,S}(t,w)\,dt.
\end{equation*}
These allocations satisfy the subtype-capacity and job-demand constraints and yield the same reward. Conversely, every feasible subtype allocation can be implemented using constant allocation fractions over the unit horizon. Thus the allocation optimum in \Cref{thm:weighted-static} equals $\OPTE_F(z)$, proving the identity for finitely supported distributions.

For general distributions, approximate each $W_S$ by a finitely supported vector $W_S^{(n)}$ through coordinatewise truncation and rounding. The integrability assumption ensures that
\begin{equation*}
\sum_{S\subseteq[m]}\lambda_S
\E\bigl[\|W_S-W_S^{(n)}\|_1\bigr]
\longrightarrow0.
\end{equation*}
This quantity bounds the change in either expected offline optimum. For matching, it bounds the expected sum of reward changes over all arrivals. For allocation, it bounds the reward change over all available type capacity. Grouping profiles with the same approximated vector gives the finite-support model, for which equality holds. Passing to the limit proves the claim.
\end{proof}

As a special case of the deterministic edge weights result, assigning a weight $w_i$ to each offline vertex corresponds to setting $w_{i,S}=w_i$ whenever $i\in S$. The objective then becomes $\sum_{i=1}^m w_i x_i$, where $x_i=\sum_{S\ni i}x_{i,S}$. This case will yield the equality of expected allocation regions below.

\subsection{Expected Allocation Regions}

The two models place randomness on different sides of the allocation problem. Matching has Poisson type supplies and unit vertex capacities. Scheduling has deterministic type capacities $\lambda_S$ and exponential service demands. The theorem shows that this transfer of randomness preserves the expected optimum for arbitrary neighborhoods. The above offline equality identity also has a geometric consequence, in terms of the expectation of random convex sets. Indeed, the two models have the same expected feasible allocation region.

For a random region $\mathcal K$, let $\E[\mathcal K]$ denote the set of vectors $\E[X]$ for all random vectors $X$ such that $X\in\mathcal K$ almost surely.\footnote{This is the expectation of random sets in the sense of Aumann \cite{Aumann1965}.} 

\begin{corollary}
\label{cor:expected-regions}
The matching and scheduling models satisfy
\begin{equation*}
\E\left[
\left\{
x\in\R_+^m:
x(A)\leq N(A),\quad\forall A\subseteq[m]
\right\}
\cap[0,1]^m
\right]
=
\E\big[
\calP\cap\prod_{i=1}^m[0,Z_i]
\big] \subset \calP\cap[0,1]^m.
\end{equation*}
Moreover, except for the trivial case $\Lambda([m])=0$, the last inclusion is strict. 
\end{corollary}

The corollary gives a coordinatewise interpretation of the duality. Every expected allocation vector achievable by an offline policy in one model is also achievable in the other. Thus transferring randomness from type supplies to vertex capacities preserves all trade-offs between the expected allocations to individual vertices. The corollary hence preserves fairness objectives defined in terms of expected allocations. For example, maximizing $\min_{i\in[m]}\E[X_i]$ or $\sum_{i=1}^m u_i(\E[X_i])$, where the $u_i$ are increasing concave utilities, gives the same optimal offline value in both models.

\begin{proof}
Fix nonnegative offline vertex weights $w_1,\ldots,w_m$. By \Cref{thm:weighted-static}, the expected maximum of $\sum_{i=1}^m w_i x_i$ over the realized feasible region is the same in both models.

For any random feasible aggregate allocation $X\in\R_+^m$, linearity gives $w\cdot\E[X]=\E[w\cdot X]$, which is at most the expected realized optimum. This bound is attained by choosing a measurable optimal allocation in each realization. Thus the two expected allocation regions have the same maximum for every nonnegative linear objective. Both expected regions are compact, convex and downward closed in $\R_+^m$. Downward closure allows coordinates with negative weights to be set to zero, so equality of maxima for nonnegative weights implies equality for arbitrary weights. The two expected regions therefore have the same support function and, being compact and convex, must coincide.

For the inclusion in the mean region, take any random feasible scheduling allocation $X$. Convexity of $\calP$ gives $\E[X]\in\calP$, while $X_i\leq Z_i$ gives $\E[X_i]\leq1$. Hence $\E[X]\in\calP\cap[0,1]^m$. If $\Lambda([m])>0$, choose $i$ with $\Lambda(\{i\})>0$. Every feasible scheduling allocation satisfies $X_i\leq\min\{\Lambda(\{i\}),Z_i\}$, so
\begin{equation*}
\E[X_i]
\leq
1-e^{-\Lambda(\{i\})}
<
\min\{\Lambda(\{i\}),1\}.
\end{equation*}
However, the vector whose $i$th coordinate is $\min\{\Lambda(\{i\}),1\}$ and whose other coordinates are zero belongs to $\calP\cap[0,1]^m$, by monotonicity of $\Lambda$. It therefore belongs to the mean region but not to the common expected allocation region, proving strict inclusion.
\end{proof}

\begin{remark}[Fluid Relaxation]
Replacing the random type supplies by their means gives the classical fluid relaxation. By \Cref{cor:expected-regions}, its value bounds the expected offline optimum:
\begin{equation*}
\OPT
\leq
\max_{x\in\calP\cap[0,1]^m}x([m]).
\end{equation*}
For a single offline vertex with arrival intensity one, the fluid benchmark equals one, while the expected offline optimum is the probability of at least one arrival, namely $1-e^{-1}$. Accepting the first arrival attains this expected optimum. Thus no online policy can guarantee a ratio greater than $1-1/e$ against the fluid benchmark, even though online decisions incur no loss in this example. This ratio is tight, attained by the Poisson version of the Suggested Matching algorithm of \cite[Section~4.1 and Appendix~C]{FeldmanEtAl2009}.

The corollary shows that replacing Poisson type supplies by their means gives a relaxation, but simultaneously replacing unit vertex capacities by independent unit-exponential capacities recovers the expected allocation region exactly. Our competitive guarantees are measured against this common expected offline optimum, for which ratios above $1-1/e$ are possible.
\end{remark}

\section{From Offline Allocations to Online Scheduling}
\label{sec:online-service}

We now prove a competitive ratio lower bound for the online scheduling problem, through some explicit policy. We prove this guarantee for an arbitrary compact, convex, downward-closed feasible-rate region $P\subseteq\R_+^m$ containing the origin and independent service demands $Z_i\sim\Exp(1)$, hence it is more general than the matching guarantee.  To do that, we construct an online scheduling policy that preserves a constant fraction of any feasible offline allocation, coordinate by coordinate. As before, $s_i(t)$ denotes the service rate of job $i$, and $x_i(t)$ its cumulative processed work.

Fix a measurable offline allocation $Y$ satisfying $Y(z)\in P$ and $0\leq Y(z)\leq z$ for every $z\in\R_+^m$. The allocation may depend on the entire demand vector. We use it to construct an online policy satisfying
\begin{equation*}
\E[x_i(1)]\geq(1-e^{1-e})\,\E[Y_i(Z)]
\qquad\forall i\in[m].
\end{equation*}
Choosing an optimal offline allocation will then give the same competitive guarantee for scheduling and Poisson matching.

\subsection{The Posterior Allocation Policy}
\label{subsec:posterior-projection}

A natural policy is to imitate a clairvoyant allocation using the information currently available. This suggests serving each unfinished job at a rate equal to the conditional expectation of its offline allocation, this corresponds to a projection of the offline allocation on the information available at time $t$. Kesselheim, Molinaro, Patton, and Singla~\cite{KesselheimEtAl2026}  develop this posterior-matching principle for several online allocation problems.

In our model, the scheduler's decisions also determine what it observes, since job completions depend on the service provided. The observed filtration therefore depends on the policy itself, so defining the rates through conditional expectations would specify the policy only implicitly. We instead define the rates directly from the observed state, using exponential memorylessness, and then prove that they equal the conditional expectations under the history generated by this policy. Completed jobs receive zero service.

At time $t$, sample an auxiliary demand vector $\widetilde Z$. For each completed job, set $\widetilde Z_i=x_i(t)$, since its demand is now known. For each unfinished job, set $\widetilde Z_i=x_i(t)+E_i$, where the $E_i$ are independent $\Exp(1)$ variables. Define
\begin{equation}
s_i(t)\coloneqq
\begin{cases}
\E[Y_i(\widetilde Z)],&\text{if job $i$ is unfinished},\\
0,&\text{if job $i$ is completed}.
\end{cases}
\label{eq:state-feedback}
\end{equation}
Here the expectation is taken over the auxiliary exponential variables, with the observed state held fixed. Every vector $Y(\widetilde Z)$ belongs to $P$, so its expectation belongs to $P$ by convexity. Setting the completed coordinates to zero preserves feasibility by downward closure. The rule therefore uses only the observed history and chooses feasible service rates.

Let $\mathcal F_t$ denote the observed service and completion history up to time $t$. The following identities show that the constructed policy is indeed the posterior projection, and how it relates to the projected demand. 

\begin{lemma}
\label{lem:adaptive-posterior}
For every unfinished job $i$,
\begin{equation}
s_i(t)=\E[Y_i(Z)\mid\mathcal F_t],
\qquad
x_i(t)+1=\E[Z_i\mid\mathcal F_t].
\label{eq:two-posterior-controls}
\end{equation}
\end{lemma}

\begin{proof}
For an unfinished job $i$, the only information revealed about its demand is that $Z_i>x_i(t)$. By memorylessness of independent exponentials, the residual demands $Z_i-x_i(t)$ of unfinished jobs are independent $\Exp(1)$ variables conditional on $\mathcal F_t$,  see~\cite[Section~2.3]{GlynnHaas2015}.

The allocation $Y_i(Z)$ may also depend on the demands of completed jobs. These are known, as $Z_j=x_j(t)$ for each completed job $j$. Thus $\widetilde Z$ has the conditional law of the full vector $Z$ given $\mathcal F_t$. The policy definition and the mean-one residual demand give, for every unfinished job $i$,
\begin{equation*}
s_i(t)\!=\E[Y_i(\widetilde Z)]
=\E[Y_i(\widetilde Z)\!\mid\mathcal F_t]
=\E[Y_i(Z)\!\mid\mathcal F_t],
\, \, \,
\E[Z_i\mid\mathcal F_t]\!=x_i(t)+\E[Z_i-x_i(t)\!\mid\mathcal F_t]=x_i(t)+1. \qedhere
\end{equation*}
\end{proof}

Subtracting the two identities gives, while job $i$ is unfinished,
\begin{equation}
x_i(t)+1-s_i(t)
=
\E[Z_i-Y_i(Z)\mid\mathcal F_t]
\geq0.
\label{eq:posterior-gap}
\end{equation}
The conditional expected demand left unallocated by $Y$ is exactly the difference between $x_i(t)+1$ and our policy's service rate. Since $x_i'(t)=s_i(t)$, while job $i$ is unfinished we have the differential inequality
\begin{equation*}
x_i'(t)\leq x_i(t)+1.
\end{equation*}
This suggests a comparison to a reference policy that solves the associated ODE $s_i(t)=x'_i(t)=x_i(t)+1$, which stops when the total work reaches $Z_i$. The induced rate need not be feasible in $P$.

To obtain a lower bound on our policy's expected work, we will control its loss relative to this reference. The identity \eqref{eq:posterior-gap} gives us the rate difference at the current state in terms of the conditional expected demand left unallocated offline. We will use it to bound the loss in expected work in terms of $\E[Z_i-Y_i(Z)]$. Since
\begin{equation*}
\E[Y_i(Z)]=1-\E[Z_i-Y_i(Z)],
\end{equation*}
this will allow us to compare the expected online work with the expected offline allocation.

\subsection{A One-Job Comparison Policy}
\label{subsec:reference-process}

Consider a job $i\in[m]$ that is unfinished at time $t$, after receiving $h$ units of work. From this time onward, serve it at a rate  $s_i(t)=x'_i(t)=x_i(t)+1$, until either the deadline or its completion. We call this the comparison policy, and recall that typically $s(t) \not \in P$ .

Let $q(t,h)$ denote the expected additional work processed by time $1$ under this policy. Thus $q(t,h)$ is the continuation value of the comparison policy, while $h+q(t,h)$ includes the work already processed. The remaining demand is unit exponential, so we can compute $q$ directly from the service prescribed by the ODE.

\begin{proposition}
\label{prop:reference-value}
For $0\leq t\leq1$ and $h\geq0$,
\begin{equation}
q(t,h)
=
1-\exp\left(-(h+1)(e^{1-t}-1)\right).
\label{eq:q-explicit}
\end{equation}
In particular, $q(1,h)=0$ and $q(0,0)=1-e^{1-e}$.
\end{proposition}

\begin{proof}
Before completion, the cumulative work $x(v)$ under the comparison policy satisfies $\frac{dx}{dv}=x+1$ with the initial condition $x(t)=h$. 
Separating variables and integrating gives
\begin{equation*}
\log\left(\frac{x(v)+1}{h+1}\right)=\int_h^{x(v)}\frac{du}{u+1}=v-t,
\end{equation*}
hence the total work is $x(v)=(h+1)e^{v-t}-1$ until completion. If the job remained unfinished until the deadline, its cumulative work would be $x(1)=(h+1)e^{1-t}-1$. Since $h$ units have already been processed, the additional work would be
\begin{equation*}
x(1)-h
=
(h+1)e^{1-t}-1-h
=
(h+1)(e^{1-t}-1).
\end{equation*}

Let $E_i\sim\Exp(1)$ denote the residual demand. The additional work is the minimum of $E_i$ and the work prescribed by the ODE. Using the tail integral formula for its expectation gives
\begin{align*}
q(t,h)=
\E\left[\min\left\{E_i,(h+1)(e^{1-t}-1)\right\}\right]
&=
\int_0^{(h+1)(e^{1-t}-1)}\Pp(E_i>u)\,du\\
&=
\int_0^{(h+1)(e^{1-t}-1)}e^{-u}\,du\\
&=
1-\exp\left(-(h+1)(e^{1-t}-1)\right).
\end{align*}
The endpoint values follow by substitution.
\end{proof}
Recall that by \Cref{lem:work-completions}, $q(t,h)$ is also the probability of completion by the deadline. Starting from time zero with no prior service, the completion probability by time $t$ is $1-e^{1-e^t}$, the same curve as the limiting matched fraction in the extremal exchangeable instance.

One key aspect of the comparison policy is that it treats each job separately. Conditional on the current state, the additional work for job $i$ depends only on its residual demand $E_i$. Its expected total additional work is therefore the sum of $q(t,x_i(t))$ over unfinished jobs. This allows us to carry out the comparison one job at a time.

We next record two properties of $q$ that will control the loss relative to the comparison policy. The first identity expresses the balance between work already processed and expected work remaining under this policy. The second bounds $1+\partial_h q(t,h)-q(t,h)$, the marginal value of serving a job now when we follow the comparison policy afterwards. To interpret this last quantity, $\varepsilon$ additional service at fixed time gives $\varepsilon$ immediate work and $\varepsilon\partial_h q(t,h)$ additional continuation value, but removes $q(t,h)$ with probability $\varepsilon$ upon job completion. 
\begin{equation*}
\varepsilon\bigl(1+\partial_h q(t,h)-q(t,h)\bigr)
+
o(\varepsilon).
\end{equation*}

\begin{lemma}
\label{lem:reference-sensitivity}
For $0\leq t\leq1$ and $h\geq0$, the function $q$ satisfies the PDE
\begin{equation}
\partial_t q(t,h)
+
(h+1)\bigl(1+\partial_h q(t,h)-q(t,h)\bigr)
=
0.
\label{eq:q-backward}
\end{equation}
As a consequence,
\begin{equation}
0\leq
1+\partial_h q(t,h)-q(t,h)
\leq
-\partial_t q(t,0).
\label{eq:q-sensitivity}
\end{equation}
\end{lemma}

\begin{proof}
Differentiating \eqref{eq:q-explicit} gives
\begin{align*}
1+\partial_h q(t,h)-q(t,h)
&=
e^{1-t}\exp\left(-(h+1)(e^{1-t}-1)\right),\\
\partial_t q(t,h)
&=
-(h+1)e^{1-t}
\exp\left(-(h+1)(e^{1-t}-1)\right).
\end{align*}
The first identity follows directly. The expression on the first line is nonnegative and nonincreasing in $h$, so it is bounded above by its value at $h=0$. By the second line, this value is $-\partial_t q(t,0)$.
\end{proof}

\subsection{Comparing the Two Policies}

Let $X_t$ denote the observed state, consisting of the cumulative work of every job and the set of unfinished jobs. Under our policy, this state is Markov: both the service rates and the conditional laws of the remaining demands are determined by the current state. For each job $i$, consider the potential
\begin{equation}
f_i(t,X_t)
\coloneqq
x_i(t)
+
\ind\{x_i(t)<Z_i\}q(t,x_i(t)).
\label{eq:comparison-potential}
\end{equation}
Conditional on $\mathcal F_t$, this is the expected total work for job $i$ if we follow our policy until time $t$ and then switch to the comparison policy. It includes the work already processed and, if the job is unfinished, the expected additional work under the comparison policy.
At time zero, the potential equals the comparison policy's expected work and at the deadline, it equals our policy's work:
\begin{equation*}
f_i(0,X_0)=q(0,0)=1-e^{1-e},
\qquad
f_i(1,X_1)=x_i(1).
\end{equation*}
Thus the expected decrease of $f_i$ between time $0$ and $1$ is exactly the loss in expected work relative to the comparison policy. We will bound this decrease by computing its instantaneous expected rate of change and integrating.

For a time-homogeneous Markov process, the \emph{infinitesimal generator} $\mathcal L$ gives the instantaneous expected rate of change of a function $f$, conditional on the current state. Given that our potential depends on $t$, we also include the time in the state. For a suitable function $f$, let
\begin{equation}
(\mathcal L f)(t,x)
\coloneqq
\lim_{\delta\to 0^+}
\frac{\E[f(t+\delta,X_{t+\delta})\mid X_t=x]-f(t,x)}{\delta}.
\label{eq:generator}
\end{equation}
Thus $(\mathcal L f)(t,x)$ is the instantaneous expected rate of change of $f(t,X_t)$, conditional on $X_t=x$. It accounts for both the explicit dependence on time and the evolution of the state. \emph{Dynkin's formula} integrates this expected rate of change \cite{Davis1984}. Under the usual regularity and integrability conditions,
\begin{equation}
\E[f(t,X_t)]-\E[f(0,X_0)]
=
\E\left[\int_0^t(\mathcal L f)(u,X_u)\,du\right].
\label{eq:dynkin}
\end{equation}
This is the analogue of the fundamental theorem of calculus for expected changes along the process.

The following lemma computes this expected rate of change for our potential under the online policy.

\begin{lemma}
\label{lem:online-generator}
Let $\mathcal L$ be the generator of the online process. While job $i$ is unfinished, with $h=x_i(t)$,
\begin{equation}
(\mathcal L f_i)(t,X_t)
=
\partial_t q(t,h)
+
s_i(t)\bigl(1+\partial_h q(t,h)-q(t,h)\bigr).
\label{eq:online-drift}
\end{equation}
After job $i$ completes, the left-hand side is zero.
\end{lemma}

\begin{proof}
Condition on $\mathcal F_t$, with job $i$ unfinished, and write $h=x_i(t)$. Over an interval of length $\delta$, bounded service rates imply that the probability of any completion is $O(\delta)$ and that of two or more completions is $O(\delta^2)$. By the exponential residual-demand law, job $i$ completes with probability $s_i(t)\delta+o(\delta)$ and does not complete with probability $1-O(\delta)$.

If no job completes, then $x_i(t+\delta)=h+s_i(t)\delta+o(\delta)$.  Since the job remains unfinished, a Taylor expansion and recalling that $x'_i(t)=s_i(t)$, gives
\begin{align*}
f_i(t+\delta,X_{t+\delta})-f_i(t,X_t)
&=
x_i(t+\delta)
+q(t+\delta,x_i(t+\delta))-h-q(t,h)\\
&=
s_i(t)\delta
+\partial_t q(t,h)\delta
+\partial_h q(t,h)s_i(t)\delta
+o(\delta)\\
&=
\delta\bigl(
\partial_t q(t,h)
+s_i(t)(1+\partial_h q(t,h))
\bigr)
+o(\delta).
\end{align*}

If job $i$ completes, its final cumulative work is $h+O(\delta)$ and its continuation value disappears, so the potential changes by $-q(t,h)+O(\delta)$. If only other jobs complete, they cause no jump in this potential; its change over the interval is $O(\delta)$, contributing only $O(\delta^2)$ to the conditional expectation. Together with the probabilities of a job completing or not completing, this yields 
\begin{equation*}
\E[f_i(t+\delta,X_{t+\delta})-f_i(t,X_t)\mid\mathcal F_t]=
\delta\bigl(
\partial_t q(t,h)
+s_i(t)(1+\partial_h q(t,h)-q(t,h))
\bigr)
+o(\delta).
\end{equation*}
Dividing by $\delta$ and letting $\delta\to0$ gives the stated identity. After job $i$ completes, its potential remains constant.
\end{proof}

The identity \eqref{eq:q-backward} shows that the generator of the comparison process, applied to $f_i$, is zero. Indeed, if we already follow the comparison policy, delaying the switch to that same policy changes nothing. The potential therefore has zero conditional expected drift. 

We can now prove the following result. 
\begin{proposition}
\label{prop:online-offline-comparison}
For every job $i\in[m]$, $\E[x_i(1)]
\geq
(1-e^{1-e})\E[Y_i(Z)]$.
\end{proposition}

\begin{proof}
Fix a job $i$. While it is unfinished, write $h=x_i(t)$. By \Cref{lem:online-generator} and \eqref{eq:q-backward}, the expected rate of change of its potential is
\begin{equation*}
(\mathcal L f_i)(t,X_t)=-\bigl(h+1-s_i(t)\bigr)
\bigl(1+\partial_h q(t,h)-q(t,h)\bigr).
\end{equation*}
By \eqref{eq:posterior-gap}, $h+1-s_i(t)
=
\E[Z_i-Y_i(Z)\mid\mathcal F_t]
\geq0$. The sensitivity bound \eqref{eq:q-sensitivity} therefore gives
\begin{equation*}
(\mathcal L f_i)(t,X_t)
= -\E[Z_i-Y_i(Z)\mid\mathcal F_t]
\bigl(1+\partial_h q(t,h)-q(t,h)\bigr)\geq 
\E[Z_i-Y_i(Z)\mid\mathcal F_t] \cdot \partial_t q(t,0).
\end{equation*}
This inequality also holds after completion: the left-hand side is then zero, while the right-hand side is nonpositive.

The service rates are bounded, there are at most $m$ completions, and $f_i$ and its derivatives along service trajectories are bounded on the states reached by time $1$. Thus Dynkin's formula applies. Using the endpoint values of $f_i$, Fubini, and the tower property,

\begin{align*}
(1-e^{1-e})-\E[x_i(1)]=\E[f_i(0,X_0)]-\E[f_i(1,X_1)]
&=-\E\left[\int_0^1(\mathcal L f_i)(t,X_t)\,dt\right]\\
&\leq \E\left[
\int_0^1(-\partial_t q(t,0))
\E[Z_i-Y_i(Z)\mid\mathcal F_t]\,dt\right]\\
&=
\int_0^1(-\partial_t q(t,0)) \cdot 
\E\bigl[\E[Z_i-Y_i(Z)\mid\mathcal F_t]\bigr]\,dt\\
&=
\E[Z_i-Y_i(Z)]\int_0^1(-\partial_t q(t,0)) \, dt\\
&=
\E[Z_i-Y_i(Z)]\bigl(q(0,0)-q(1,0)\bigr)\\
&=
(1-e^{1-e})\E[Z_i-Y_i(Z)].
\end{align*}
Finally, $\E[Z_i]=1$, so rearranging gives
\begin{equation*}
\E[x_i(1)]
\geq
(1-e^{1-e})\bigl(1-\E[Z_i-Y_i(Z)]\bigr)
=
(1-e^{1-e})\E[Y_i(Z)]. \qedhere
\end{equation*}
\end{proof}
The coordinate-wise guarantee also extends to ex-ante concave objectives, as in \cite{duetting2026fairprophets}. For any nonnegative, coordinatewise nondecreasing concave function $f$ and feasible offline allocation $Y$, monotonicity and concavity give
\begin{equation*}
f\bigl(\E[x(1)]\bigr)\geq f\left(\left(1-\frac{e}{e^e}\right)\E[Y(Z)]\right)\geq\left(1-\frac{e}{e^e}\right)f\bigl(\E[Y(Z)]\bigr).
\end{equation*}
Choosing $Y$ to optimize this objective yields the same competitive guarantee. By \Cref{lem:online-equivalence,cor:expected-regions}, this applies both to expected work allocated to individual jobs and to matching probabilities of individual offline vertices.

We now specialize to $P=\calP$ and choose a measurable optimal offline allocation $Y(Z)$, so that $Y(Z)\in\calP$, $Y(Z)\leq Z$, and $Y(Z)([m])=\OPTE(Z)$. 

\begin{corollary}
\label{cor:poisson-competitive}
For every Poisson matching instance with known arrival intensities, there exists an online policy satisfying
\begin{equation*}
\ALG\geq(1-e^{1-e})\OPT.
\end{equation*}
\end{corollary}

\begin{proof}
Applying \Cref{prop:online-offline-comparison} to the optimal offline allocation and summing over jobs gives 
\begin{equation*}
\ALG=\E\left[\sum_{i=1}^m x_i(1)\right]
\geq
(1-e^{1-e})\,\E[Y(Z)([m])]
=
(1-e^{1-e})\,\E[\OPTE(Z)]
=
(1-e^{1-e})\OPT.
\end{equation*}
The first equality is \Cref{lem:online-equivalence}, and the last equality is by \Cref{thm:offline-duality}.
\end{proof}

\section{Transfer with Exactly $n$ i.i.d. Arrivals and Tightness}
\label{sec:fixed-count}

We transfer the Poisson guarantee to exactly $n$ i.i.d.~arrivals from a known distribution. Huang and Shu~\cite{HuangShu2021} compared the offline benchmarks using concavity and obtained an online transfer with an $O(n^{-1/2})$ loss for $O(1)$-approximately monotone policies. Yan~\cite[Appendix~A.1]{Yan2026JailletLu} removed this assumption, including for nonnegative edge weights, with an $O(n^{-1/3})$ loss. We refine Yan's simulation by removing its truncation and using a direct sampling comparison of the offline benchmarks. This gives an explicit $O(n^{-1/2})$ loss for arbitrary policies, also with nonnegative edge weights. The benchmark comparison follows from concavity of the expected maximum matching, as done by Huang and Shu~\cite[proof of Theorem~1]{HuangShu2021}. We include a direct sampling proof for completeness.

Fix $n\geq1$ and a known distribution $\nu$ on finitely many arrival types, each specifying a neighborhood and nonnegative edge weights. Let $\OPT_k$ denote the expected maximum matching weight for $k$ independent samples from $\nu$, with $\OPT_0=0$. Consider a rate-$n$ Poisson process on $[0,1]$ whose arrival types are independent samples from $\nu$. Its total arrival count is $K\sim\Pois(n)$, and its expected offline optimum is $\OPT=\E[\OPT_K]$. 

\begin{lemma}
\label{lem:uniform-position-transfer}
For matching with nonnegative edge weights, the offline benchmarks satisfy
\begin{equation}
\OPT\geq
\left(1-e^{-n}\frac{n^n}{n!}\right)\OPT_n.
\label{eq:poissonized-benchmark}
\end{equation}
Moreover, every Poisson online policy with expected reward $\ALG$ can be converted into an online policy for exactly $n$ i.i.d.~arrivals whose expected reward satisfies
\begin{equation}
\ALG_n\geq\ALG-e^{-n}\frac{n^n}{n!}\OPT_n.
\label{eq:uniform-position-transfer}
\end{equation}
\end{lemma}

\begin{proof}
\textbf{Offline transfer.} For each realization of $k$ arrivals, do the following. Take a maximum-weight matching on $k$ arrivals, and select $\ell\leq k$ of these arrivals uniformly at random. Retaining only the matching edges incident to the selected arrivals gives a feasible matching on the sampled $\ell$ arrivals. Each edge is retained with probability $\ell/k$, so this matching has expected weight $(\ell/k)\OPT_k$. Since the selected arrivals are still $\ell$ independent samples from $\nu$, their expected offline optimum is $\OPT_\ell$. Therefore,
\begin{equation*}
\OPT_\ell\geq\frac{\ell}{k}\OPT_k.
\end{equation*}
Together with monotonicity, this gives
\begin{equation*}
\OPT_k\geq\min\{k/n,1\}\OPT_n,
\qquad
\OPT_k\leq\frac{k}{n}\OPT_n\quad\text{for }k\geq n.
\end{equation*}
The Poisson recursion $k\Pp(K=k)=n\Pp(K=k-1)$ gives
\begin{align*}
\E[(n-K)^+]
&=\E[(K-n)^++n-K]\\
&=\E[(K-n)^+]+n-\E[K]
=\E[(K-n)^+]\\
&=\sum_{k=n+1}^{\infty}(k-n)\Pp(K=k)\\
&=n\sum_{k=n+1}^{\infty}\bigl[\Pp(K=k-1)-\Pp(K=k)\bigr]\\
&=n\Pp(K=n)
=ne^{-n}\frac{n^n}{n!}.
\end{align*}
Hence
\begin{align*}
\OPT
\geq\E\left[\min\left\{\frac{K}{n},1\right\}\right]\OPT_n
=\E\left[1-\frac{(n-K)^+}{n}\right]\OPT_n
=\left(1-e^{-n}\frac{n^n}{n!}\right)\OPT_n.
\end{align*}
\medskip

\textbf{Online transfer.} Fix a Poisson online policy with expected reward $\ALG$ when the total arrival rate is $n$. We construct a policy for exactly $n$ i.i.d.~arrivals using Yan's simulation~\cite[Appendix~A.1]{Yan2026JailletLu}, without truncation. Independently of the real arrivals, sample a rate-$n$ Poisson process on $[0,1]$, with total count $K$ and ordered event times $T_1<\cdots<T_K$. We run the Poisson policy on this virtual time scale. The simulated policy observes only the virtual arrival history, it is not given $K$, future event times, or the real/dummy labels in advance.

If $K\leq n$, then upon the $i$-th real arrival, for $i\leq K$, advance the virtual time to $T_i$ and present this arrival to the simulated policy. Implement its matching decision and ignore the remaining $n-K$ real arrivals. If $K>n$, generate $K-n$ additional 'fake' arrivals independently from $\nu$. Choose a uniformly random $n$-subset of the $K$ virtual positions, independently of the event times and arrival types. Place the real arrivals in these positions in their original order, and the additional simulated arrivals in the remaining positions. Run the Poisson policy on this combined sequence at the sampled event times. This produces a feasible matching on the real and simulated arrivals. Removing all edges incident to simulated arrivals leaves a feasible matching on the real arrivals. The construction can be realized online because, upon each real arrival, we process any preceding simulated arrivals and immediately implement the policy's decision for the real arrival. Simulated matches reserve vertices only within the simulation.

The simulated input has the same law as the original Poisson input, so the simulated matching has expected weight $\ALG$. Conditional on $K=k>n$, the combined arrival types are i.i.d.~and independent of the positions chosen for the real arrivals. Since the Poisson policy does not observe which arrivals are real, its matching is also independent of these positions. Each matched edge therefore survives deletion with probability $n/k$. The conditional expected weight of the simulated matching is at most $\OPT_k$, so the conditional expected loss is at most $(1-n/k)\OPT_k$. There is no loss when $K\leq n$. Using $\OPT_k\leq(k/n)\OPT_n$ for $k\geq n$ and averaging gives
\begin{align*}
\ALG-\ALG_n
&\leq\sum_{k>n}\Pp(K=k)\left(1-\frac nk\right)\OPT_k\\
&\leq\sum_{k>n}\Pp(K=k)\frac{k-n}{n}\OPT_n\\
&=\frac{\E[(K-n)^+]}{n}\OPT_n
=e^{-n}\frac{n^n}{n!}\OPT_n. \qedhere
\end{align*}
\end{proof}

This lets us transfer a competitive ratio established in the Poisson model to the i.i.d. model at a small loss. 
We now specialize to unweighted matching. Let $\mathsf{CR}_n$ denote the optimal worst-case competitive ratio for exactly $n$ i.i.d.~arrivals from a known distribution, over all offline vertex sets and arrival distributions.

\begin{corollary}
\label{cor:fixed-count-guarantee}
For every $n\geq1$,
\begin{equation*}
\mathsf{CR}_n\geq1-\frac{e}{e^e}-\left(2-\frac{e}{e^e}\right)e^{-n}\frac{n^n}{n!}\geq1-\frac{e}{e^e}-\frac{2-e/e^e}{\sqrt{2\pi n}}.
\end{equation*}
In particular,
\begin{equation*}
\lim_{n\to\infty}\mathsf{CR}_n=1-\frac{e}{e^e}.
\end{equation*}
\end{corollary}

\begin{proof}
Fix an instance with exactly $n$ i.i.d.~arrivals and consider its Poissonization with total rate $n$. By \Cref{cor:poisson-competitive}, there is a Poisson policy satisfying $\ALG\geq(1-e/e^e)\OPT$. Applying both inequalities of \Cref{lem:uniform-position-transfer}, the corresponding fixed-count policy satisfies
\allowdisplaybreaks
\begin{align*}
\ALG_n
&\geq\ALG-e^{-n}\frac{n^n}{n!}\OPT_n\\
&\geq\left(1-\frac{e}{e^e}\right)\OPT-e^{-n}\frac{n^n}{n!}\OPT_n\\
&\geq\left[\left(1-\frac{e}{e^e}\right)\left(1-e^{-n}\frac{n^n}{n!}\right)-e^{-n}\frac{n^n}{n!}\right]\OPT_n\\
&=\left[1-\frac{e}{e^e}-\left(2-\frac{e}{e^e}\right)e^{-n}\frac{n^n}{n!}\right]\OPT_n.
\end{align*}
Since this guarantee holds for every instance, it proves the first bound on $\mathsf{CR}_n$. Stirling's inequality $n!\geq\sqrt{2\pi n}(n/e)^n$ gives $e^{-n}n^n/n!\leq1/\sqrt{2\pi n}$, and hence
\begin{equation*}
\mathsf{CR}_n\geq1-\frac{e}{e^e}-\frac{2-e/e^e}{\sqrt{2\pi n}}.
\end{equation*}
Letting $n\to\infty$ yields
\begin{equation*}
\liminf_{n\to\infty}\mathsf{CR}_n\geq1-\frac{e}{e^e}.
\end{equation*}
Chierichetti, Giacchini, Panconesi, and  Vattani~\cite{ChierichettiEtAl2025} construct instances with $n$ i.i.d.~arrivals from a known distribution on which every online policy satisfies
\begin{equation*}
\frac{\ALG_n}{\OPT_n}\leq1-\frac{e}{e^e}+o(1).
\end{equation*}
This gives the matching upper bound on $\limsup_{n\to\infty}\mathsf{CR}_n$ and proves the stated limit.
\end{proof}

We can at the same time finish the proof of the main theorem. 

\begin{proof}[Proof of \Cref{thm:CR-poisson}]
The same transfer establishes sharpness in the known-rate Poisson model. For any online policy on the Poisson instances corresponding to the construction of Chierichetti, Giacchini, Panconesi, and  Vattani~\cite{ChierichettiEtAl2025}, \Cref{lem:uniform-position-transfer} gives
\begin{equation*}
\frac{\ALG}{\OPT}\leq\frac{\ALG_n/\OPT_n+e^{-n}n^n/n!}{1-e^{-n}n^n/n!}\leq1-\frac{e}{e^e}+o(1).
\end{equation*}
Together with \Cref{cor:poisson-competitive}, this shows that the optimal universal competitive ratio in the known-rate Poisson model is $1-e/e^e$. 
\end{proof}
Remark that to prove the tight competitive ratio for Poisson arrivals, the previously known transfer results were sufficient. Here the main advantage is keeping the improved $n^{-1/2}$ loss for all policies.

\section{Not Knowing the Arrival Law is Strictly Harder}
\label{sec:unknown-law}

We now show that knowing the arrival law strictly improves the optimal competitive ratio in the Poisson model. We construct a prior over arrival laws for which the offline optimum matches almost every vertex, but an online policy cannot identify a small class of vertices early enough to give them priority.

\begin{theorem}
\label{thm:unknown-poisson}
In the Poisson arrival model, there exists a constant $\varepsilon>0$ such that
\begin{equation*}
\mathsf{CR}^{\text{unknown}}
\coloneqq\sup_{\ALG}\inf_{\boldsymbol{\lambda}\geq0}\frac{\ALG}{\OPT}
\leq1-\frac{e}{e^e}-\varepsilon<1-\frac{e}{e^e}.
\end{equation*}
\end{theorem}
The infimum ranges over finite Poisson instances with $\OPT>0$.
The upper bound holds even when the policy knows the total rate and the prior over arrival laws.

Recall the generating function of the worst-case exchangeable degree distribution of Chierichetti, Giacchini, Panconesi, and  Vattani~\cite{ChierichettiEtAl2025}:
\begin{equation*}
\Psi(z)\coloneqq z+(1-z)\log(1-z)=\sum_{d=2}^{\infty}\frac{z^d}{d(d-1)},
\qquad 0\leq z<1.
\end{equation*}

\paragraph{Construction of the hard instance.}
By the superposition and thinning properties of Poisson processes, specifying the neighborhood rates $\lambda_S$ is equivalent to choosing a total arrival rate and independently sampling a neighborhood at each arrival. Fix a rational $s\in(0,1)$ and choose $m$ such that $s^2m$ is an integer. For each set $H\subseteq[m]$ of size $s^2m$, we construct a fixed Poisson instance as follows. We choose the total arrival rate to be $m$, so that the expected number of arrivals over $[0,1]$ equals the number of offline vertices. Each arrival independently draws an integer $D$ with $\Pp(D=d)=1/[d(d-1)]$ for $d\geq2$. With probability $s^D$, the arrival draws $D$ vertices independently and uniformly from $H$, with replacement, and its neighborhood consists of the distinct vertices sampled. Otherwise, its neighborhood is $[m]$. We call these arrivals restricted and universal, respectively. For each fixed $H$, arrivals are independently classified as restricted with probability $\sum_{d=2}^{\infty}s^d/[d(d-1)]=\Psi(s)$ and universal otherwise. By Poisson thinning, the restricted and universal arrivals therefore form independent Poisson processes with rates $m\Psi(s)$ and $m(1-\Psi(s))$.

We place a uniform prior on this family by choosing $H$ uniformly once, before any arrivals, and using the corresponding instance throughout. The policy knows the family and the prior, but not the chosen $H$. Conditional on $H$, arrival neighborhoods are independent and identically distributed. Averaging over $H$ makes them dependent, since past neighborhoods provide information about $H$ and hence about future neighborhoods. \medskip 

The policy can recognize restricted arrivals because their neighborhoods are proper subsets of $[m]$. Each such arrival reveals that its neighbors belong to $H$, but leaves the membership of all other vertices unresolved. Universal arrivals provide no information about $H$. If $H$ were known, the policy could prioritize vertices in $[m]\setminus H$ on universal arrivals, since these vertices cannot be matched by restricted arrivals. We show that a policy knowing only the prior cannot identify these vertices quickly enough.
Throughout the argument, we will take $m\to\infty$ with $s$ fixed. We will then choose one rational $s$ sufficiently close to one. \medskip

We first bound the offline optimum for each fixed $H$. We recall that $\OPT$ is the expected offline performance averaged over  $H$. 
\begin{lemma}
\label{lem:unknown-offline}
For fixed $s$, the expected offline optimum is identical for every choice of $H$ and satisfies
\begin{equation*}
\lim_{m \to \infty} \frac{\OPT}{m} =1.
\end{equation*}
\end{lemma}

\begin{proof}
Fix $H$ and replace the restriction probability $s^D$ by $s^2$, which makes restricted arrivals with smaller neighborhoods more likely. The restricted arrivals then form the Poissonized instance of \cite[Definition~8]{ChierichettiEtAl2025} on $s^2m$ offline vertices, with total arrival rate $s^2m$. By \cite[Theorem~14]{ChierichettiEtAl2025} and the offline benchmark comparison in \Cref{lem:uniform-position-transfer}, their expected offline optimum is $s^2m-o(m)$.

The universal arrivals have rate $(1-s^2)m$ and can be matched to distinct vertices of $[m]\setminus H$ until either these vertices or the universal arrivals are exhausted. Their expected contribution is therefore $(1-s^2)m-O(\sqrt m)$. Combining this matching with the matching within $H$ gives expected size $m-o(m)$.

Since $D\geq2$ implies $s^D\leq s^2$, we can couple by stochastic dominance the modified instance with the original one by replacing some restricted neighborhoods by $[m]$. This only adds edges, so the original instance also has expected offline optimum $m-o(m)$. Any two choices of $H$ are related by a permutation of the offline vertices, so their expected offline optima are identical. Averaging over the uniform prior therefore gives $\OPT=m-o(m)$.
\end{proof}

The next lemma bounds the time integral of the expected fraction of vertices that an online policy can match outside $H$. This integral will control the error term in the differential comparison below.

\begin{lemma}
\label{lem:unknown-learning}
For an online policy, let $I(t)$ denote the expected number of matched vertices in $[m]\setminus H$ by time $t$, divided by $m$. Uniformly over online policies, as $s \to 1$,
\begin{equation*}
\limsup_{m\to\infty}\int_0^1 I(t)\,dt=o(1-s).
\end{equation*}
\end{lemma}

\begin{proof}
Let $U(t)$ denote the number of vertices that have not appeared in any restricted arrival by time $t$. \medskip

\textbf{Law of unseen vertices outside of $H$.} Every neighborhood of size strictly smaller than $m$ identifies a restricted arrival, and all its vertices belong to $H$. Thus, conditional on the history up to time $t$, $H$ must contain the $m-U(t)$ vertices already seen in restricted arrivals. Its remaining $s^2m-(m-U(t))$ vertices must lie among the $U(t)$ unseen vertices.

Since the prior on $H$ is uniform and the restricted arrival law is exchangeable, the remaining vertices of $H$ form a uniformly chosen subset of the unseen vertices, conditional on the observed history. Consequently, each unseen vertex lies outside $H$ with conditional probability
\begin{equation*}
\frac{U(t)-\bigl(s^2m-(m-U(t))\bigr)}{U(t)}=\frac{(1-s^2)m}{U(t)}.
\end{equation*} 
\medskip 

\textbf{Relating $U(t)$ and $I(t)$.} Universal arrivals form a Poisson process of rate $m(1-\Psi(s))$, independent of $H$ and the restricted arrivals, and hence of $U(t)$. Let $M(t)$ denote the number of universal arrivals by time $t$.

Every matched unseen vertex must have been matched through a universal arrival, so there are at most $M(t)$ such vertices. Conditional on the history, each lies outside $H$ with probability $(1-s^2)m/U(t)$ as proved above. By linearity of conditional expectation, taking expectations over the history and dividing by $m$ gives
\begin{equation*}
I(t)\leq(1-s^2)\E\left[\frac{M(t)}{U(t)}\right]=(1-s^2)(1-\Psi(s))t\,\E\left[\frac{m}{U(t)}\right],
\end{equation*}
where the equality uses independence of $M(t)$ and $U(t)$ and $\E[M(t)]=m(1-\Psi(s))t$. Since there are only $(1-s^2)m$ vertices outside $H$, we also have $I(t)\leq1-s^2$. Combining the two bounds yields
\begin{equation}
\frac{I(t)}{1-s^2}\leq\min\left\{1,(1-\Psi(s))t\,\E\left[\frac{m}{U(t)}\right]\right\}.
\label{eq:unknown-unseen-bound}
\end{equation}
\medskip

\textbf{Computing $\E [U(t)]$.}
We now compute the mean of $U$. We will see later on that because $U$ concentrates, it is sufficient.  
Fix $H$ and a vertex $v\in H$. An arrival draws $D=d$ with probability $1/[d(d-1)]$ and is then restricted with probability $s^d$. Conditional on these events, each of its $d$ independent samples avoids $v$ with probability $(s^2m -1)/(s^2m)$. Thus
\begin{equation*}
\Pp(D=d,\ \text{arrival is restricted and contains }v\mid H)=\frac{1}{d(d-1)}\,s^d\left[1-\left(1-\frac{1}{s^2m}\right)^d\right].
\end{equation*}
Summing over $d$ gives the probability that an arrival is restricted and contains $v$. By Poisson thinning, such arrivals form a Poisson process whose rate is $m$ times this probability,
\begin{align*}
m\sum_{d=2}^{\infty}\frac{s^d}{d(d-1)}\left[1-\left(1-\frac{1}{s^2m}\right)^d\right]=m\left[\Psi(s)-\Psi\left(s-\frac{1}{sm}\right)\right]\xrightarrow[m \to \infty]{}\frac{\Psi'(s)}{s}=-\frac{\log(1-s)}{s}.
\end{align*}
The vertex $v$ remains unseen at time $t$ exactly when no restricted arrival containing it has occurred. A Poisson process of rate $r$ has no arrivals by time $t$ with probability $e^{-rt}$. Hence
\begin{align*}
\Pp(v\text{ is unseen at time }t\mid H)\xrightarrow[m \to \infty]{}\exp\left(\frac{t}{s}\log(1-s)\right)=(1-s)^{t/s}.
\end{align*}
All $(1-s^2)m$ vertices outside $H$ remain unseen, since restricted arrivals only sample from $H$. Each of the $s^2m$ vertices inside $H$ remains unseen with the same probability, which converges to $(1-s)^{t/s}$. Thus, for any $v\in H$, linearity of expectation gives
\begin{equation*}
\frac{\E[U(t)\mid H]}{m}=1-s^2+s^2\,\Pp(v\text{ is unseen at time }t\mid H)\xrightarrow[m\to\infty]{}1-s^2+s^2(1-s)^{t/s}.
\end{equation*}
The conditional mean is identical for every $H$, so the same limit holds for $\E[U(t)]/m$. \medskip

\textbf{Variance of $U(t)$.} 
We will show that $\operatorname{Var}(U(t))=O_s(m)$. Fix $H$ and two distinct vertices $v,w\in H$. Conditional on an arrival being restricted and having $D=d$, for both $v$ and $w$ to appear, one of the $d$ draws must return $v$ and a different draw must return $w$. There are $d(d-1)$ choices of these two draws. Since sampling is independent and with replacement, each choice has probability $1/(s^2m)^2$. This gives by union bound
\begin{equation*}
\Pp(v,w\text{ both appear}\mid H,D=d,\text{restricted})\leq\frac{d(d-1)}{(s^2m)^2}.
\end{equation*}
The rate $q$ of restricted arrivals containing both vertices therefore satisfies
\begin{equation*}
q\leq m\sum_{d=2}^{\infty}\frac{s^d}{d(d-1)}\frac{d(d-1)}{(s^2m)^2}
=\frac{1}{s^4m}\sum_{d=2}^{\infty}s^d
=\frac{1}{s^2m(1-s)}
=O_s(1/m).
\end{equation*}

Let $r$ be the rate of restricted arrivals containing any given vertex of $H$. By symmetry, this rate is the same for every vertex and depends only on $m$ and $s$. Restricted arrivals containing at least one of $v,w$ have rate $2r-q$. Since both vertices remain unseen exactly when no such arrival occurs,
\begin{equation*}
\Pp(v\text{ unseen at }t\mid H)=e^{-rt},
\qquad
\Pp(v,w\text{ both unseen at }t\mid H)=e^{-(2r-q)t}.
\end{equation*}
Hence,
\begin{align*}
\operatorname{Cov}\bigl(\mathbf{1}_{\{v\text{ unseen at }t\}},\mathbf{1}_{\{w\text{ unseen at }t\}}\mid H\bigr)
=e^{-(2r-q)t}-e^{-2rt}
=e^{-2rt}(e^{qt}-1)
\leq e^q-1
=O_s(1/m),
\end{align*}
uniformly for $t\in[0,1]$. Every vertex outside $H$ remains unseen, so $U(t)=(1-s^2)m+\sum_{v\in H}\mathbf{1}_{\{v\text{ unseen at }t\}}$.
Summing the individual variances and the pairwise covariances therefore gives
\begin{align*}
\operatorname{Var}(U(t)\mid H)
&=\sum_{v\in H}\operatorname{Var}\bigl(\mathbf{1}_{\{v\text{ unseen at }t\}}\mid H\bigr)\quad+\sum_{\substack{v,w\in H\\v\neq w}}\operatorname{Cov}\bigl(\mathbf{1}_{\{v\text{ unseen at }t\}},\mathbf{1}_{\{w\text{ unseen at }t\}}\mid H\bigr)\\
&\leq \vert H \vert +O_s\big(\frac{|H|^2}{m}\big)
=O_s(m).
\end{align*}
Moreover, $\E[U(t)\mid H]$ does not depend on $H$ as seen above. Thus the conditional mean has zero variance over the prior, and the law of total variance yields
\begin{align*}
\operatorname{Var}(U(t))=\E[\operatorname{Var}(U(t)\mid H)]
+\operatorname{Var}(\E[U(t)\mid H])=\E[\operatorname{Var}(U(t)\mid H)]
=O_s(m).
\end{align*}

\textbf{Combining the bounds.} Since $U(t)\geq(1-s^2)m$ as we will never see vertices outside of $H$ in the restricted arrivals, Cauchy--Schwarz gives
\begin{align*}
\left|\E[U(t)]\E[\frac{1}{U(t)}]-1\right|
=\left|\E[\frac{\E[U(t)]-U(t)}{U(t)}]\right|
\leq\frac{\E[|U(t)-\E[U(t)]|]}{(1-s^2)m}
\leq\frac{\sqrt{\operatorname{Var}(U(t))}}{(1-s^2)m}
=O_s(\frac{1}{\sqrt{m}})=o(1).
\end{align*}
Thus $\E[1/U(t)]=(1+o(1))/\E[U(t)]$ as $m\to\infty$ with $s$ fixed.
Combining the concentration estimate with the expression for the mean gives
\begin{equation*}
\E\left[\frac{m}{U(t)}\right]\xrightarrow[m\to\infty]{}\frac{1}{1-s^2+s^2(1-s)^{t/s}}.
\end{equation*}
Since $I(t)\leq1-s^2$, we have $0\leq I(t)/(1-s)\leq1+s\leq2$. The reverse Fatou inequality therefore gives
\begin{align*}
\limsup_{m\to\infty}\frac{1}{1-s}\int_0^1 I(t)\,dt\leq\int_0^1\limsup_{m\to\infty}\frac{I(t)}{1-s}\,dt\leq(1+s)\int_0^1\min\left\{1,\frac{(1-\Psi(s))t}{1-s^2+s^2(1-s)^{t/s}}\right\}dt.
\end{align*}
For fixed $t<1$, the exponent $1-t/s$ converges to $1-t>0$ as $s\to1$. In particular, $1-t/s\geq(1-t)/2$ for $s$ sufficiently close to one. Using $1-\Psi(s)=(1-s)(1-\log(1-s))$ and dropping the nonnegative term $1-s^2$ from the denominator gives
\begin{align*}
0\leq\frac{(1-\Psi(s))t}{1-s^2+s^2(1-s)^{t/s}}\leq\frac{t}{s^2}(1-s)^{1-t/s}(1-\log(1-s))\leq\frac{t}{s^2}(1-s)^{(1-t)/2}(1-\log(1-s))
\xrightarrow[s\to1]{}0.
\end{align*}
The last limit follows because $(1-t)/2>0$ and any positive power of $1-s$ tends to zero faster than the logarithmic factor grows. 

The integrand, including the factor $1+s$, is bounded by two. Applying the reverse Fatou inequality again as $s\to1$ shows that the right-hand side has limit superior at most zero. By nonnegativity, it therefore converges to zero, proving the claim uniformly over policies.
\end{proof}

We can now finish the upper bound proof. 

\begin{proof}[Proof of \Cref{thm:unknown-poisson}]
Fix an online policy and let $A(t)$ denote the expected total number of matched vertices by time $t$, divided by $m$. For the analysis, condition on $H$ and the history up to time $t$, and suppose that $k$ vertices of $H$ are currently matched. \medskip

\textbf{A differential inequality for $A$.} The probability that the next arrival is restricted and has all its neighbors matched is
\begin{equation*}
\sum_{d=2}^{\infty}\frac{s^d}{d(d-1)}\left(\frac{k}{s^2m}\right)^d=\Psi\left(\frac{k}{sm}\right).
\end{equation*}
Such an arrival cannot be matched. Averaging over $H$ and the history, the expected number of matched vertices in $H$, divided by $m$, is $A(t)-I(t)$. These $k$ vertices remain matched, so restricted arrivals whose neighbors all lie among them remain unmatchable throughout $(t,t+h]$. Since arrivals occur at rate $m$, the conditional expected number of new matches in this interval is at most $mh(1-\Psi(k/(sm)))$. Averaging over $H$ and the history and dividing by $m$ gives
\begin{equation*}
A(t+h)-A(t)\leq h\,\E\left[1-\Psi\left(\frac{k}{sm}\right)\right].
\end{equation*}
Moreover, at most one match is made per arrival, so $0\leq A(t+h)-A(t)\leq h$. Thus $A$ is Lipschitz and differentiable almost everywhere. Dividing the preceding inequality by $h$, letting $h\downarrow0$, and then applying Jensen's inequality, using the convexity of $\Psi$, yields, almost everywhere,
\begin{align*}
A'(t)\leq\E\left[1-\Psi\left(\frac{k}{sm}\right)\right]\leq1-\Psi\left(\frac{\E[k]}{sm}\right)=1-\Psi\left(\frac{A(t)-I(t)}{s}\right).
\end{align*}

\textbf{Comparison with a reference trajectory.}
Recall the exchangeable trajectory, defined for $t\geq0$ by
\begin{equation*}
x(t)\coloneqq1-e^{1-e^t},
\qquad x'(t)=1-\Psi(x(t)),
\qquad x(0)=0.
\end{equation*}
The function $sx(t/s)$ solves $y'=1-\Psi(y/s)$ with initial condition $y(0)=0$, as
\begin{equation*}
\frac{d}{dt}\bigl(sx(t/s)\bigr)=x'(t/s)=1-\Psi(x(t/s)),
\end{equation*}
and $sx(0)=0$. We compare $A(t)$ with this trajectory and use \Cref{lem:unknown-learning} to control the contribution of $I(t)$. 

The convexity inequality for $\Psi$ at $x(t/s)$ gives
\begin{equation*}
\Psi\left(\frac{A(t)-I(t)}{s}\right)\geq\Psi(x(t/s))+\frac{\Psi'(x(t/s))}{s}\bigl(A(t)-sx(t/s)-I(t)\bigr).
\end{equation*}
Subtracting the derivative of $sx(t/s)$ from the bound on $A'(t)$ therefore yields
\begin{equation*}
\frac{d}{dt}\bigl(A(t)-sx(t/s)\bigr)= A'(t) -\big(1-\Psi(x(t/s))\big) \leq-\frac{\Psi'(x(t/s))}{s}\bigl(A(t)-sx(t/s)-I(t)\bigr).
\end{equation*}
Both trajectories start at zero. 
Multiplying by $\exp(\int_0^t\Psi'(x(u/s))/s\,du)$, applying the product rule, and using the previous inequality gives
\begin{align*}
&\frac{d}{dt}\left[\exp\left(\int_0^t\frac{\Psi'(x(u/s))}{s}\,du\right)\bigl(A(t)-sx(t/s)\bigr)\right]\\
&\quad=\exp\left(\int_0^t\frac{\Psi'(x(u/s))}{s}\,du\right)
\left[\frac{d}{dt}\bigl(A(t)-sx(t/s)\bigr)+\frac{\Psi'(x(t/s))}{s}\bigl(A(t)-sx(t/s)\bigr)\right]\\
&\quad\leq\exp\left(\int_0^t\frac{\Psi'(x(u/s))}{s}\,du\right)\frac{\Psi'(x(t/s))}{s}I(t).
\end{align*}
The two terms proportional to $A(t)-sx(t/s)$ cancel in the last step. Integrating from $0$ to $1$, using the zero initial difference, and dividing by $\exp(\int_0^1\Psi'(x(u/s))/s\,du)$ yields
\begin{equation*}
A(1)-sx(1/s)\leq\int_0^1\exp\left(-\int_t^1\frac{\Psi'(x(u/s))}{s}\,du\right)\frac{\Psi'(x(t/s))}{s}I(t)\,dt.
\end{equation*}
Since $\Psi'(x(t/s))=-\log(1-x(t/s))=-\log\left(e^{1-e^{t/s}}\right)=e^{t/s}-1$ is nonnegative and increasing in $t$, the exponential factor is at most one and $\Psi'(x(t/s))/s\leq\Psi'(x(1/s))/s$. Using $I(t)\geq0$, we conclude that
\begin{equation*}
A(1)\leq sx(1/s)+\frac{\Psi'(x(1/s))}{s}\int_0^1 I(t)\,dt.
\end{equation*}
The coefficient of the integral converges to $\Psi'(x(1))=e-1$ as $s\to 1$. By \Cref{lem:unknown-learning},
\begin{equation*}
\limsup_{m\to\infty}\frac{\ALG}{m}\leq sx(1/s)+o(1-s),
\end{equation*}
uniformly over online policies.

The function $x$ is strictly concave on $(0,\infty)$ and satisfies $x(0)=0$. Its tangent at time $1$ therefore gives
\begin{equation*}
sx(1/s)\leq x(1)-(1-s)\bigl(x(1)-x'(1)\bigr),
\qquad x(1)-x'(1)>0.
\end{equation*}
Choose one rational $s<1$ sufficiently close to one and keep it fixed. Since $x(1)=1-e/e^e$, we obtain, for some $\varepsilon>0$,
\begin{equation*}
\ALG\leq\left(1-\frac{e}{e^e}-2\varepsilon\right)m+o(m),
\qquad \OPT=m-o(m).
\end{equation*}
For all sufficiently large $m$, the ratio under the prior is therefore at most $1-e/e^e-\varepsilon$, uniformly over online policies.

The expected offline optimum is identical for every choice of $H$. Averaging over $H$ therefore implies that, for each online policy, there exists some fixed choice of $H$ that gives a competitive ratio no larger than the ratio under the prior. This proves the bound on $\mathsf{CR}^{\text{unknown}}$. 
\end{proof}

The same separation holds asymptotically in the fixed-count i.i.d.~model. For each fixed $H$, take exactly $n=m$ independent arrivals from the same neighborhood law. Run any policy for this fixed-count model on the Poisson instance, stopping after $n$ arrivals or at time $1$, whichever occurs first. This defines an online policy without requiring knowledge of $H$. Conditional on $H$, coupling the executions on the same neighborhood sequence gives
\begin{equation*}
\ALG_n\leq\ALG+\E[(n-K)^+],
\qquad |\OPT_n-\OPT|\leq\E[|K-n|],
\qquad K\sim\Pois(n),
\end{equation*}
where the expectations also average over the prior. Both error terms are $O(\sqrt n)$. Since $\OPT=n-o(n)$, we obtain $\OPT_n=n-o(n)$ and the same strict asymptotic gap. Symmetry and averaging over $H$ again give a fixed hard arrival law for each policy.

\section*{Acknowledgments}

This project has been partially funded by the European Research Council (ERC) under the European Union's Horizon Europe Program (FACT, grant agreement No.~101170373), by an Amazon Research Award, and by the Israel Science Foundation Breakthrough Program (grant No.~2600/24).

\newpage

\appendix

\section{AI Use Disclosure}

We made substantial use of ChatGPT and Gemini as interactive research assistants. The authors supplied several guiding ideas, including polymatroid representations of allocations, the use of functional and stochastic analysis, a reformulation of the offline benchmark (which was motivated by a similar approach in \cite{loiseau2026} building on earlier work by some of the authors), and a hidden-subset construction to separate known and unknown arrival laws. AI assistance contributed to developing and testing these directions, formulating candidate lemmas and algorithms, searching for counterexamples and relevant literature, conducting numerical experiments, and generating verification code. The tools also assisted with drafting and revising the manuscript. The process was iterative: proposed arguments were evaluated and corrected or discarded when gaps or errors were found.

The authors wrote the manuscript, with AI assistance in refining its style and presentation. They developed the paper organization, proof structure and exposition, critically assessed proposed formulations, and revised and reformulated mathematical arguments. They checked the proofs, literature attributions, and claims of novelty, and take full responsibility for the work.

\section{Weighted Offline Equivalence}
\label{app:weighted-offline}

We prove the weighted offline equality
\begin{equation*}
\E\bigl[\OPTP_w(\mathbf N)\bigr]
=
\E\bigl[\OPTE_w(Z)\bigr].
\end{equation*}
As before, the offline vertices are indexed by $[m]$, the independent counts satisfy $N_S\sim\mathrm{Pois}(\lambda_S)$, and the independent capacities satisfy $Z_i\sim\mathrm{Exp}(1)$. An allocation of one unit of type $S$ to vertex $i\in S$ has weight $w_{i,S}\geq0$. We first express both benchmarks as minima over vertex prices. An Euler-characteristic identity reduces the comparison of these minima to comparisons on individual simplices. We then construct a common triangulation whose vertices form price chains, along which the count--waiting-time identity from the unweighted proof applies.

\subsection{Offline Optima and Dual Prices}

The deficiency representations in the unweighted proof are the
unit-weight case of standard bipartite matching duality, 
see Schrijver~\cite[Sections~17.1 and~21.1]{Schrijver2003}.
Eliminating the type-side dual variables gives the 
representations below in terms of vertex prices. For nonnegative vertex prices $p\in\mathbb R_+^m$, define the residual type price
\begin{equation*}
b_S(p)=\max_{i \in S} (w_{i,S}-p_i)_+
\end{equation*}
and the two convex objectives
\begin{equation*}
\Phi_{\mathrm P}(p)=\sum_{i=1}^m p_i+\sum_S N_Sb_S(p),
\qquad
\Phi_{\mathrm E}(p)=\sum_{i=1}^m Z_ip_i+\sum_S\lambda_Sb_S(p).
\end{equation*}

\begin{lemma}
\label{lem:weighted-dual-prices}
For every realization of $\mathbf N$ and $Z$,
\begin{equation*}
\OPTP_w(\mathbf N)=\min_{p\geq0}\Phi_{\mathrm P}(p),
\qquad
\OPTE_w(Z)=\min_{p\geq0}\Phi_{\mathrm E}(p).
\end{equation*}
Both minima may be restricted to $[0,W]^m$ whenever $W$ bounds all edge weights.
\end{lemma}

\begin{proof}
Recall the primal matching LP: 
\begin{equation*}
\begin{aligned}
\max\quad &\sum_S\sum_{i\in S}w_{i,S}x_{i,S}\\
\text{subject to}\quad &\sum_{S\ni i}x_{i,S}\leq1 &&\text{for every }i,\\
&\sum_{i\in S}x_{i,S}\leq N_S &&\text{for every }S,\\
&x_{i,S}\geq0 &&\text{for every }S\text{ and }i\in S.
\end{aligned}
\end{equation*}
Its constraint matrix is the incidence matrix of a bipartite graph and is therefore totally unimodular. Since its capacities are integers, the LP has an integral optimal solution. Each such solution assigns at most one arrival to each offline vertex and uses at most $N_S$ arrivals of each type $S$, so the LP optimum equals $\OPTP_w(\mathbf N)$.

Assign nonnegative multipliers $p_i$ to the vertex-capacity constraints and $q_S$ to the type-capacity constraints. The Lagrangian is
\begin{align*}
\sum_S\sum_{i\in S}w_{i,S}x_{i,S}+\sum_i p_i\left(1-\sum_{S\ni i}x_{i,S}\right)+\sum_S q_S\left(N_S-\sum_{i\in S}x_{i,S}\right)\\
=\sum_i p_i+\sum_S N_Sq_S+\sum_S\sum_{i\in S}(w_{i,S}-p_i-q_S)x_{i,S}.
\end{align*}
For feasible $x$, the added terms are nonnegative, so this expression bounds the primal objective from above. After relaxing the capacity constraints, its supremum over $x\geq0$ is finite exactly when $p_i+q_S\geq w_{i,S}$ for every $S$ and $i\in S$. Otherwise, a variable with positive coefficient can be increased without bound. When these inequalities hold, the supremum is $\sum_i p_i+\sum_S N_Sq_S$. The primal LP is feasible and bounded, so strong LP duality gives
\begin{equation*}
\OPTP_w(\mathbf N)=\min\left\{\sum_i p_i+\sum_S N_Sq_S:p_i+q_S\geq w_{i,S}\ \text{for }i\in S,\quad p,q\geq0\right\}.
\end{equation*}
For fixed $p\geq0$, the smallest feasible choice is $q_S=\max_{i\in S}(w_{i,S}-p_i)_+=b_S(p)$, with $q_{\varnothing}=0$. Since $N_S\geq0$, this choice minimizes the dual objective over $q$, giving $\OPTP_w(\mathbf N)=\min_{p\geq0}\Phi_{\mathrm P}(p)$.

The exponential allocation benchmark is the optimum of the same primal LP with vertex capacities $Z_i$ and type capacities $\lambda_S$. This LP is also feasible and bounded. The same calculation and strong duality give the same dual constraints, with objective $\sum_i Z_ip_i+\sum_S\lambda_Sq_S$. Minimizing over $q$ again gives $q_S=b_S(p)$ and proves the second formula. Finally, capping any coordinate $p_i$ at $W$ leaves every $b_S(p)$ unchanged and cannot increase either objective.
\end{proof}

At each fixed price vector,
\begin{equation}
\label{eq:weighted-fixed-price}
\E[\Phi_{\mathrm P}(p)]
=\sum_{i=1}^m p_i+\sum_S\lambda_Sb_S(p)
=\E[\Phi_{\mathrm E}(p)].
\end{equation}
The minimizing prices depend on the realization, however, so we still need to compare the expected minima.

\subsection{Convex Minima through Simplices}

The next identity follows from the Euler-integration formulas of Baryshnikov and Ghrist~\cite[Lemma~1 and Proposition~2]{BaryshnikovGhrist2010}, and leveraging convexity of the objective. We give instead an elementary proof using the same Euler-characteristic principle and simplicial complex topology  as in the unweighted argument.\medskip 

A simplex is the convex hull of finitely many affinely independent points. A triangulation of a polytope $B$ is a realization of a finite simplicial complex by Euclidean simplices whose union is $B$. Concretely, it is a finite collection $\mathcal T$ of simplices containing all their nonempty faces, such that any two simplices intersect either in a common face or not at all, and their union is $B$. We identify the faces of the simplicial complex with their corresponding geometric simplices.

For example, label the vertices of the square $[0,1]^2$ by $a=(0,0)$, $b=(1,0)$, $c=(1,1)$, and $d=(0,1)$. Drawing the diagonal from $a$ to $c$ gives a triangulation whose associated simplicial complex has maximal faces $\{a,b,c\}$ and $\{a,c,d\}$, together with all their subsets. Geometrically, these faces correspond to two closed triangles, five distinct edges (including the diagonal), and four vertices. The triangles cover the square and intersect along the diagonal, which realizes their common face $\{a,c\}$.

\begin{lemma}
\label{lem:convex-euler}
Let $\mathcal T$ be a finite triangulation of a nonempty compact convex polytope $B$. If $g:B\to\mathbb R$ is convex and affine on every simplex, then
\begin{equation*}
\min_{x\in B}g(x)
=
\sum_{\varnothing\neq\sigma\in\mathcal T}
(-1)^{\dim\sigma}
\max_{x\in\sigma}g(x).
\end{equation*}
\end{lemma}

\begin{proof}
For a set $B$ admitting a finite triangulation $\mathcal T$, we define its Euler characteristic as that of the underlying simplicial complex:
\begin{equation*}
\chi(B):=\chi(\mathcal T)
=\sum_{\varnothing\neq\sigma\in\mathcal T}(-1)^{\dim\sigma}.
\end{equation*}
We use the standard facts that this definition is independent of the chosen triangulation and that homotopy-equivalent spaces have the same Euler characteristic.\footnote{Two spaces are homotopy equivalent if there are continuous maps between them such that each composition can be continuously deformed into the identity map on the corresponding space. In particular, a deformation retraction gives a homotopy equivalence.} A nonempty contractible space has Euler characteristic $1$, and we set $\chi(\varnothing)=0$.
Since $B$ is nonempty and convex, it is contractible, and therefore
\begin{equation*}
\sum_{\varnothing\neq\sigma\in\mathcal T}(-1)^{\dim\sigma}
=\chi(B)=1.
\end{equation*}
Adding a constant $c$ to $g$ adds $c$ to its minimum and adds $c\sum_{\sigma\neq\varnothing}(-1)^{\dim\sigma}=c$ to the proposed right-hand side. Since $g$ is continuous on the compact set $B$, we may thus shift it by a constant and assume $g\geq0$.
For $t\geq0$, define
\begin{equation*}
B_t=\{x\in B:g(x)\leq t\},
\qquad
\mathcal T_t=\{\sigma\in\mathcal T:\sigma\subseteq B_t\}.
\end{equation*}
The collection $\mathcal T_t$ is a subcomplex, since every face of a simplex contained in $B_t$ is also contained in $B_t$. Its union need not equal $B_t$, because the level $t$ may cut through some simplices. If $B_t$ is empty, $\mathcal T_t$ has no nonempty simplices and $\chi(\mathcal T_t)=0$. Suppose now that $B_t$ is nonempty. In that case, because $g$ is convex, $B_t$ is contractible, and $\chi(B_t)=1$. But the union of $\mathcal T_t$ need not be convex. We show that nevertheless the union of $\mathcal T_t$ remains contractible. \medskip 

\textbf{Constructing a continuous projection onto $\mathcal T_t$.} First, we construct a continuous map $r_t$ from $B_t$ onto this union that fixes the union pointwise. This map will be continuous in $x \in B_t$ for $t$ fixed, such that when we move between adjacent simplexes, it moves continuously. For $x\in\sigma\cap B_t$, let $v_0,\ldots,v_d$ be the vertices of $\sigma$ and write the unique barycentric representation
\begin{equation*}
x=\sum_{j=0}^d\alpha_jv_j,
\qquad
\alpha_j\geq0,
\qquad
\sum_{j=0}^d\alpha_j=1.
\end{equation*}
Because $g$ is affine and the weights sum to $1$, this gives $g(x)=\sum_j\alpha_jg(v_j)\leq t$. Consequently, at least one vertex with $g(v_j)\leq t$ has positive coefficient. We can therefore discard the coordinates of vertices with value greater than $t$ and renormalize the remaining coordinates:
\begin{equation*}
r_t(x)=
\frac{\displaystyle\sum_{j:\,g(v_j)\leq t}\alpha_jv_j}
{\displaystyle\sum_{j:\,g(v_j)\leq t}\alpha_j}.
\end{equation*}
This point lies in the face of $\sigma$ spanned by the vertices with value at most $t$. Every point of this face is a convex combination of vertices whose values under $g$ are at most $t$. By affineness, $g\leq t$ throughout the face, so the face lies in $B_t$ and belongs to $\mathcal T_t$.

The formula is continuous on each $\sigma\cap B_t$, since its denominator is positive there. On shared faces, the barycentric coordinates at shared vertices agree, while all other coordinates vanish, so the formulas for $r_t$ coincide. They therefore define a continuous map on $B_t$, by gluing the continuous maps on the finitely many closed sets $\sigma\cap B_t$. If $x$ lies in a simplex of $\mathcal T_t$, all its vertices have value at most $t$, so no coordinates are discarded and $r_t(x)=x$. \medskip

\textbf{Constructing a continuous deformation.}
For $x\in B_t$, consider the continuous deformation
\begin{equation*}
(x,s)\longmapsto (1-s)x+s\,r_t(x),
\qquad s\in[0,1].
\end{equation*}
Both $x$ and $r_t(x)$ lie in $B_t$, which is convex by convexity of $g$, so this deformation stays in $B_t$. At $s=0$ it is the identity, and at $s=1$ it maps $B_t$ onto the union of $\mathcal T_t$. Moreover, every point already in this union remains unchanged throughout, since $r_t(x)=x$ there. Thus the union of $\mathcal T_t$ is a deformation retract of $B_t$, and the two spaces are homotopy equivalent. By homotopy invariance of Euler characteristic and contractibility of nonempty convex sets, 
\begin{equation*}
\chi(\mathcal T_t)=\chi(B_t)=\mathbf1\!\left\{B_t \neq \emptyset\right\}.
\end{equation*}

\textbf{Integrating the Euler Characteristic.}
A simplex $\sigma$ belongs to $\mathcal T_t$ exactly when $\max_{x\in\sigma}g(x)\leq t$, as \(g\) must have value at most \(t\) throughout the simplex, while $B_t$ is nonempty exactly when $\min_{x\in B}g(x)\leq t$. Combining the empty and nonempty cases gives
\begin{equation*}
\sum_{\varnothing\neq\sigma\in\mathcal T}(-1)^{\dim\sigma}
\mathbf1\!\left\{\max_{x\in\sigma}g(x)\leq t\right\}
=
\chi(\mathcal T_t)
=
\mathbf1\!\left\{\min_{x\in B}g(x)\leq t\right\}.
\end{equation*}
Subtracting this equality from $\sum_{\sigma\neq\varnothing}(-1)^{\dim\sigma}=1$ yields
\begin{equation*}
\mathbf1\!\left\{t<\min_{x\in B}g(x)\right\}
=
\sum_{\varnothing\neq\sigma\in\mathcal T}(-1)^{\dim\sigma}
\mathbf1\!\left\{t<\max_{x\in\sigma}g(x)\right\}.
\end{equation*}
Finally, integrate over $t\geq0$. Since $g\geq0$, we may apply $\int_0^\infty\mathbf1\{t<a\}\,dt=a$ to every extremum above. By linearity of the integral, this gives
\begin{equation*}
\min_{x\in B}g(x)
=
\sum_{\varnothing\neq\sigma\in\mathcal T}(-1)^{\dim\sigma}
\max_{x\in\sigma}g(x). \qedhere
\end{equation*}
\end{proof}

The next section will construct a triangulation in the price space. 

\subsection{Poisson--Exponential Duality on Price Chains}

Assume first that the weights are integers, and choose an integer $W\geq1$ bounding all weights. By Lemma~\ref{lem:weighted-dual-prices}, it suffices to consider prices in $[0,W]^m$. Divide this domain into unit cubes. In each cube $a+[0,1]^m$, use the simplices with vertices
\begin{equation}
\label{eq:weighted-chain-triangulation}
a,\quad a+e_{\pi(1)},\quad
a+e_{\pi(1)}+e_{\pi(2)},\quad\ldots,\quad a+\mathbf1,
\end{equation}
where $e_i$ is the $i$-th coordinate vector and $\pi$ ranges over all permutations of $[m]$. 
These simplices agree on shared cube faces, and together with all their faces, they form a triangulation $\mathcal T$ of $[0,W]^m$. For example, when $m=2$ and $[0,1]^2$, the triangulation corresponds to the two triangles obtained by drawing the diagonal from $(0,0)$ to $(1,1)$.

As a first intermediary result, we prove that $\Phi_{\mathrm P}$ and $\Phi_{\mathrm E}$ are appropriately affine. 
\begin{lemma}\label{lem:potential-affine}
The functions $\Phi_{\mathrm P}$ and $\Phi_{\mathrm E}$ are affine on every simplex $\sigma\in\mathcal T$.
\end{lemma}

\begin{proof}
We first show that every $b_S$ is affine on each simplex of $\mathcal T$. Fix a nonempty type $S$ and a unit cube $a+[0,1]^m$. Writing a price vector in this cube as $p=a+x$, we have $b_S(a+x)=\max_{i\in S}(w_{i,S}-a_i-x_i)_+$.
Let $h=\max_{i\in S}(w_{i,S}-a_i)$, which is an integer. If $h\leq0$, then every term inside the positive part is nonpositive, so $b_S$ vanishes throughout the cube.

Suppose now that $h\geq1$. Since the values $w_{i,S}-a_i$ are integers, any value smaller than $h$ is at most $h-1$. Subtracting $x_i\geq0$ can only decrease it. In contrast, an index satisfying $w_{i,S}-a_i=h$ gives $h-x_i\geq h-1\geq0$, because $x_i\leq1$. Thus the maximum defining $b_S(a+x)$ is always attained among the indices with $w_{i,S}-a_i=h$, and the positive part is unnecessary for these indices. Consequently,
\begin{equation*}
b_S(a+x)
=\max_{\substack{i\in S\\w_{i,S}-a_i=h}}(h-x_i)
=h-\min_{\substack{i\in S\\w_{i,S}-a_i=h}}x_i.
\end{equation*}

We have a triangulation of the unit cube, so in particular, $a+x$ lies in one such simplex $\sigma$. Let $\pi$ be the order associated to this simplex. At each vertex $v \in \sigma$,  the coordinates of $v$ minus $a$ are a sequence of $1$s and then $0$s, so in particular:
\begin{equation*}
v_{\pi(1)}-a_{\pi(1)}\geq\cdots\geq v_{\pi(m)}-a_{\pi(m)}.
\end{equation*}
Because these inequalities hold at every vertex, they hold throughout their convex hull by convex combination. Therefore for every $a+x \in \sigma$, 
\begin{equation*}
x_{\pi(1)}\geq\cdots\geq x_{\pi(m)}.
\end{equation*}
Hence, the indices satisfying $w_{i,S}-a_i=h$ are fixed throughout the cube. Among these indices, the minimum above is therefore always given by the coordinate appearing last in the order $\pi$. Hence $b_S(a+x)$ has the form $h-x_i$ for a fixed index $i$ on this simplex, which is an affine function of $p=a+x$. 

For $S=\varnothing$, we have $b_S=0$. Since each objective is a linear function of the prices plus a linear combination of the functions $b_S$, both $\Phi_{\mathrm P}$ and $\Phi_{\mathrm E}$ are affine on every simplex of $\mathcal T$, for every realization.    
\end{proof}

We next adapt the Poisson-process argument from the unweighted case to compare the expected maxima of the two objectives on each simplex, by using the chains induced by prices.
\begin{lemma}
\label{lem:weighted-price-chains}
For every nonempty simplex $\sigma\in\mathcal T$,
\begin{equation*}
\E\!\left[\max_{p\in\sigma}\Phi_{\mathrm P}(p)\right]
=
\E\!\left[\max_{p\in\sigma}\Phi_{\mathrm E}(p)\right].
\end{equation*}
\end{lemma}

\begin{proof}
Fix a nonempty simplex $\sigma\in\mathcal T$ of dimension $k\leq m$, and write its vertices in coordinatewise order as $p^0,\ldots,p^k$. If $k=0$, the claim follows from equality of the expected objectives at every fixed price vector. Assume henceforth that $k\geq1$.

The vertices of $\sigma$ belong to one of the chains used to triangulate a unit cube. Consequently, each coordinate increases by either $0$ or $1$ between $p^0$ and any later vertex. We may therefore write
\begin{equation*}
p^r=p^0+\mathbf1_{A_r},
\qquad
\varnothing=A_0\subsetneq A_1\subsetneq\cdots\subsetneq A_k,
\end{equation*}
where $A_r$ is the set of coordinates that have increased by $1$ from $p^0$ to $p^r$. Here $\mathbf1_{A_r}$ denotes the vector equal to $1$ on $A_r$ and $0$ elsewhere. A step may increase several coordinates, since $\sigma$ may be a face that omits some vertices of the full cube chain.

As the prices increase along this chain, each $b_S$ is nonincreasing. Moreover, since every coordinate increases by at most $1$ in total, every surplus $w_{i,S}-p_i$ decreases by at most $1$, and hence
\begin{equation*}
b_S(p^0)-1\leq b_S(p^r)\leq b_S(p^0).
\end{equation*}
All weights and vertex prices are integers, so these values of $b_S$ are integers. Thus, along the chain, $b_S$ either stays constant or drops by exactly $1$ at a single step.

For each $r$, let
\begin{equation*}
J_r=\{S:b_S(p^r)=b_S(p^0)-1\}
\end{equation*}
be the collection of types whose residual price has already dropped by the time we reach $p^r$. Since these prices cannot increase again, the sets $J_r$ are nested:
\begin{equation*}
\varnothing=J_0\subseteq J_1\subseteq\cdots\subseteq J_k.
\end{equation*}

By definition of $A_r$, the difference $p_i^r-p_i^0$ equals $1$ when $i\in A_r$ and $0$ otherwise. Similarly, each residual price either stays unchanged or drops by $1$, and $J_r$ consists exactly of the types for which this drop has occurred. Thus $b_S(p^r)-b_S(p^0)$ equals $-1$ when $S\in J_r$ and $0$ otherwise. Substituting these differences into the definition of $\Phi_{\mathrm P}$ gives
\begin{equation*}
\Phi_{\mathrm P}(p^r)-\Phi_{\mathrm P}(p^0)=\sum_{i=1}^m(p_i^r-p_i^0)
+\sum_S N_S\bigl(b_S(p^r)-b_S(p^0)\bigr)=\sum_{i\in A_r}1+\!\!\sum_{S\in J_r}N_S(-1)=|A_r|-\!\!\sum_{S\in J_r}N_S.
\end{equation*}
Similarly, for the exponential dual objective, we get
\begin{equation*}
\Phi_{\mathrm E}(p^r)-\Phi_{\mathrm E}(p^0)=\sum_{i\in A_r}Z_i-\sum_{S\in J_r}\lambda_S. 
\end{equation*}

We use the Poisson-process argument from the proof of Lemma~\ref{lem:chain-duality}. Define the observation times
\begin{equation*}
t_r:=\sum_{S\in J_r}\lambda_S,
\qquad r=0,\ldots,k.
\end{equation*}
Since the sets $J_r$ are nested, these times satisfy $0=t_0\leq t_1\leq\cdots\leq t_k$. Let $\Pi(t)$ count the arrivals up to time $t$ in an auxiliary rate-one Poisson process, and let $T_j$ denote its $j$th arrival time, with $T_0=0$. The nested type counts have increments
\begin{equation*}
\sum_{S\in J_r}N_S-\sum_{S\in J_{r-1}}N_S
=\sum_{S\in J_r\setminus J_{r-1}}N_S
\sim\Pois(t_r-t_{r-1}).
\end{equation*}
These increments are independent because the groups $J_r\setminus J_{r-1}$ are disjoint. They therefore have the same joint law as the increments $\Pi(t_r)-\Pi(t_{r-1})$. Since both sequences start at zero,
\begin{equation*}
\bigg(\sum_{S\in J_r}N_S\bigg)_{r=0}^k
\stackrel{d}{=}
\bigl(\Pi(t_r)\bigr)_{r=0}^k.
\end{equation*}
For the exponential sums, the arrival-time representation established in the proof of Lemma~\ref{lem:chain-duality} gives
\begin{equation*}
\bigl(Z(A_r)\bigr)_{r=0}^k
\stackrel{d}{=}
\bigl(T_{|A_r|}\bigr)_{r=0}^k.
\end{equation*}
Indeed, both sequences start at zero and have independent increments, each consisting of $|A_r|-|A_{r-1}|$ independent rate-one exponential variables.

Set
\begin{equation*}
M_{\mathrm P}
:=\max_{0\leq r\leq k}\bigg\{|A_r|-\sum_{S\in J_r}N_S\bigg\},
\qquad
M_{\mathrm E}
:=\max_{0\leq r\leq k}\{Z(A_r)-t_r\}.
\end{equation*}
These are the maxima of the objective increments displayed above. Both are nonnegative because the term $r=0$ is zero. Fix $a\geq0$, and let $K_a\sim\Pois(a)$ be independent of both $\Pi$ and $\mathbf N$. As in the proof of Lemma~\ref{lem:chain-duality}, stationary independent increments give
\begin{equation*}
\bigl(\Pi(a+t_r)\bigr)_{r=1}^k
\stackrel{d}{=}
\bigl(K_a+\Pi(t_r)\bigr)_{r=1}^k.
\end{equation*}
To see this, write $\Pi(a+t_r)=\Pi(a)+[\Pi(a+t_r)-\Pi(a)]$. The bracketed terms have the same joint law as $(\Pi(t_r))_{r=1}^k$ and are independent of $\Pi(a)\sim\Pois(a)$. The same variable $K_a$ therefore appears at every observation time. Using the equivalence $T_j>u$ if and only if $\Pi(u)<j$, we obtain
\begin{align*}
\Pp(M_{\mathrm E}>a)
&=\Pp\left(Z(A_r)>t_r+a
\ \text{for some }r\in[k]\right)\\
&=\Pp\left(T_{|A_r|}>t_r+a
\ \text{for some }r\in[k]\right)\\
&=\Pp\left(\Pi(t_r+a)<|A_r|
\ \text{for some }r\in[k]\right)\\
&=\Pp\left(K_a+\Pi(t_r)<|A_r|
\ \text{for some }r\in[k]\right)\\
&=\Pp\left(K_a+\sum_{S\in J_r}N_S<|A_r|
\ \text{for some }r\in[k]\right)\\
&=\Pp(M_{\mathrm P}>K_a).
\end{align*}
Here the representations above are joint distribution identities along the entire chain, which allows us to apply them to the event that at least one inequality holds.

Since $M_{\mathrm P}$ is a nonnegative integer bounded by $m$, independence of $K_a$ and $\mathbf N$ gives
\begin{align*}
\E[M_{\mathrm E}]=\int_0^\infty\Pp(M_{\mathrm E}>a)\,da
=\sum_{j=0}^{m-1}\Pp(M_{\mathrm P}>j)
\int_0^\infty e^{-a}\frac{a^j}{j!}\,da
=\sum_{j=0}^{m-1}\Pp(M_{\mathrm P}>j)
=\E[M_{\mathrm P}],
\end{align*}
where each integral equals $1$, and the last equality is the tail-sum formula for a nonnegative integer-valued random variable. Finally, the objective increment identities give
\begin{equation*}
\max_{0\leq r\leq k}\Phi_{\mathrm P}(p^r)
=\Phi_{\mathrm P}(p^0)+M_{\mathrm P},
\qquad
\max_{0\leq r\leq k}\Phi_{\mathrm E}(p^r)
=\Phi_{\mathrm E}(p^0)+M_{\mathrm E}.
\end{equation*}
Since the two objectives have the same expectation at the fixed price vector $p^0$, we conclude that
\begin{equation*}
\E\left[\max_{0\leq r\leq k}\Phi_{\mathrm P}(p^r)\right]
=
\E\left[\max_{0\leq r\leq k}\Phi_{\mathrm E}(p^r)\right].
\end{equation*}
Finally, $\Phi_{\mathrm E}$ and $\Phi_{\mathrm P}$ are both convex. Therefore, their maximum over a simplex is attained at its vertices, i.e. 
$\max_{0\leq r\leq k}\Phi_{\mathrm E}(p^r)=\max_{p \in \sigma}\Phi_{\mathrm E}(p)$ and $\max_{0\leq r\leq k}\Phi_{\mathrm P}(p^r)=\max_{p \in \sigma}\Phi_{\mathrm P}(p)$. Combining this with the previous expectation equality concludes the proof. 
\end{proof}

\subsection{From Price Chains to Weighted Offline Equality}

We can now combine the Euler identity, the fact that the potential functions are affine on the triangulation, and the equality in expectation of the maximum over simplices to finish the proof. 

\begin{proof}[Proof of the weighted offline equality]
For integer weights, Lemmas~\ref{lem:weighted-dual-prices}, \ref{lem:convex-euler}, \ref{lem:potential-affine}, and~\ref{lem:weighted-price-chains} give
\begin{align*}
\E\bigl[\OPTP_w(\mathbf N)\bigr]
&=\E\!\left[\min_{p\in[0,W]^m}\Phi_{\mathrm P}(p)\right]\\
&=\sum_{\varnothing\neq\sigma\in\mathcal T}(-1)^{\dim\sigma}
  \E\!\left[\max_{p\in\sigma}\Phi_{\mathrm P}(p)\right]\\
&=\sum_{\varnothing\neq\sigma\in\mathcal T}(-1)^{\dim\sigma}
  \E\!\left[\max_{p\in\sigma}\Phi_{\mathrm E}(p)\right]\\
&=\E\!\left[\min_{p\in[0,W]^m}\Phi_{\mathrm E}(p)\right]
 =\E\bigl[\OPTE_w(Z)\bigr].
\end{align*}
All sums are finite, and all terms are integrable. \medskip

For rational weights, multiply by a common denominator and use homogeneity.  For arbitrary nonnegative real weights, approximate uniformly by rational vectors.  If $w,w'$ are two weight vectors, then
\[
|\OPTP_w(\mathbf N)-\OPTP_{w'}(\mathbf N)|
\le m\|w-w'\|_\infty,
\]
because any matching uses at most $m$ edges, while
\[
|\OPTE_{w}(Z)-\OPTE_{w'}(Z)|
\le Z([m])\|w-w'\|_\infty.
\]
Taking expectations, using $\E Z([m])=m$, and using the density of rationals in the reals proves Theorem~\ref{thm:weighted-static}. 
\end{proof}

\section{The Offline Limit for Exchangeable Arrivals}
\label{app:exchangeable-limit}

We derive the offline limit used in \Cref{subsec:exchangeable-warmup} from the bipartite matching formula of Bordenave, Lelarge, and Salez~\cite[Theorem~3]{BordenaveLelargeSalez2013}. We first treat bounded degree distributions, then remove this restriction by truncation and pass to Poisson arrivals.

Let $D$ be a nonnegative integer-valued random variable with generating function $\phi(u)\coloneqq\E[u^D]$. Consider $m$ offline vertices and Poisson arrivals of total rate $m$ over $[0,1]$. Each arrival independently draws $D$ and samples $D$ neighbors independently and uniformly from $[m]$, with replacement. We have the following limit equality.
\begin{proposition}
\label{prop:exchangeable-offline-limit}
\begin{equation*}
\lim_{m\to\infty}\frac{\E[\OPTP(\mathbf N)]}{m}
=1-\max_{u\in[0,1]}\{\phi(u)-\Psi(u)\},
\end{equation*}
where $\Psi(u)\coloneqq u+(1-u)\log(1-u)$ for $0\leq u<1$ and $\Psi(1)\coloneqq1$. 
\end{proposition}

\begin{proof}
First consider exactly $m$ independent arrivals, and let $\OPT_m$ denote their expected offline optimum. Suppose that $D$ is bounded and write $\mu\coloneqq\E[D]=\phi'(1)$. If $\mu=0$, every neighborhood is empty and the claimed formula is immediate, so assume $\mu>0$. The resulting bipartite graph has $m$ arrival vertices and $m$ offline vertices. \medskip

\textbf{Limit degree laws.} 
Conditional on $D=d$, each pair of distinct draws returns the same vertex with probability $1/m$. A union bound over the $\binom d2$ pairs, followed by averaging over $D$, gives
\begin{equation*}
\Pp(\text{an arrival has fewer than }D\text{ distinct neighbors})
\leq\frac{\E[D(D-1)]}{2m}
=O(1/m),
\end{equation*}
where the last bound uses that $D$ is bounded. Thus the arrival degree equals $D$ with probability tending to one, so its limiting distribution has generating function $\phi$.

Fix an offline vertex $v$. Conditional on $D=d$, each sample avoids $v$ with probability $1-1/m$. Since the samples are independent, all $d$ samples avoid $v$ with probability $(1-1/m)^d$. Averaging over $D$ and taking the complementary event gives
\begin{equation*}
\Pp(\text{arrival contains }v)
=1-\E\left[\left(1-\frac1m\right)^D\right]
=1-\phi\left(1-\frac1m\right),
\end{equation*}
where the last equality uses $\phi(u)=\E[u^D]$. Thus the degree of an offline vertex has distribution
\begin{equation*}
\operatorname{Bin}\bigl(m,1-\phi(1-1/m)\bigr)\xrightarrow[m\to\infty]{\mathrm d}\Pois(\mu),
\end{equation*} 
because $m[1-\phi(1-1/m)]=m[\phi(1)-\phi(1-1/m)]\to \phi'(1)=\mu$. \medskip

\textbf{Weak local limit.} Intuitively, as  $m$ grows large, the degree of offline and online nodes is basically bounded with high probability, so the graph is sparse, and locally looks like a tree. This can be made formal as follows.

Local weak convergence describes the neighborhood of a uniformly chosen vertex. More precisely, for each fixed radius and finite rooted graph, the proportion of vertices whose neighborhood has that rooted shape converges to the corresponding probability under the limiting rooted graph. For our random graphs, these proportions converge in probability.

Retain repeated neighbor samples as parallel edges. Conditional on the degree sequences on both sides, the resulting multigraph is a bipartite configuration model, since the arrival-side half-edges are paired uniformly with the offline-side half-edges. By the law of large numbers and standard occupancy estimates, its empirical degree distributions converge in probability to the law of $D$ on the arrival side and $\Pois(\mu)$ on the offline side, and their first moments converge to $\mu$. The standard local convergence result for the bipartite configuration model therefore gives the hierarchical Galton--Watson limit with these degree distributions, see \cite[Section~5.2]{BordenaveLelarge2010}. Its root lies on either side with probability $1/2$, since both sides have $m$ vertices. Collapsing parallel edges preserves this local limit and does not change the maximum matching size.

We can thus apply \cite[Theorem~3]{BordenaveLelargeSalez2013}. The correspondence between their notation and ours is
\begin{center}
\begin{tabular}{ll}
Their notation & Our setting \\ \hline
$V^a$ & The $m$ arrival vertices \\
$V^b$ & The $m$ offline vertices \\
$\pi^a$ & The law of $D$ \\
$\pi^b$ & $\Pois(\mu)$ \\
$\phi^a(u)$ & $\phi(u)$ \\
$\phi^b(u)$ & $e^{\mu(u-1)}$ \\
$\lambda$ & $1/2$
\end{tabular}
\end{center}
Both degree distributions have mean $\mu$, so the balance condition in that theorem holds. Their matching density is normalized by the total number $2m$ of vertices. Adjusting to our normalization by $m$, and using boundedness of the normalized matching size to pass to expectations, the theorem gives
\begin{equation*}
\lim_{m\to\infty}\frac{\OPT_m}{m}
=1-\max_{t\in[0,1]}\left\{
\phi(1-e^{-\mu t})-\bigl(1-e^{-\mu t}-\mu t e^{-\mu t}\bigr)
\right\}.
\end{equation*}
Indeed, substituting $\phi^a=\phi$ and $\phi^b(u)=e^{\mu(u-1)}$ into their expression for $F^a$ gives
\begin{align*}
F^a(t)
=\phi^a\left(1-\frac{\phi^{b\prime}(1-t)}{\phi^{b\prime}(1)}\right)
-\frac{\phi^{a\prime}(1)}{\phi^{b\prime}(1)}
\bigl(1-\phi^b(1-t)-t\phi^{b\prime}(1-t)\bigr)=\phi(1-e^{-\mu t})-\bigl(1-e^{-\mu t}-\mu t e^{-\mu t}\bigr).
\end{align*}

Under the change of variable $u=1-e^{-\mu t}$, we have
\begin{equation*}
\Psi(u)=u+(1-u)\log(1-u)=1-e^{-\mu t}-\mu t e^{-\mu t}.
\end{equation*}
Thus
\begin{equation*}
\lim_{m\to\infty}\frac{\OPT_m}{m}
=1-\max_{u\in[0,1-e^{-\mu}]}\{\phi(u)-\Psi(u)\}.
\end{equation*}
To extend the maximum to $[0,1]$, observe that by convexity $\phi'(u)\leq\phi'(1)=\mu$. For $1-e^{-\mu}<u<1$,
\begin{equation*}
\frac{d}{du}\bigl(\phi(u)-\Psi(u)\bigr)
=\phi'(u)+\log(1-u)
\leq\mu+\log(1-u)<0.
\end{equation*}
Thus $\phi-\Psi$ is decreasing on the omitted interval. By continuity at $u=1$, extending the interval does not change the maximum, and
\begin{equation*}
\lim_{m\to\infty}\frac{\OPT_m}{m}
=1-\max_{u\in[0,1]}\{\phi(u)-\Psi(u)\}.
\end{equation*}

\textbf{Truncation argument.} Now let $D$ be arbitrary. For an integer $L\geq1$, replace $D$ by $\min\{D,L\}$ and write $\phi_L(u)\coloneqq\E[u^{\min\{D,L\}}]$. Couple the instances by retaining only the first $L$ neighbor samples of each arrival. Only arrivals with $D>L$ are changed, and deleting all matching edges incident to these arrivals loses at most one edge per changed arrival. Writing $\OPT_m^{(L)}$ for the expected optimum of the truncated instance, we obtain
\begin{equation*}
0\leq\frac{\OPT_m-\OPT_m^{(L)}}{m}\leq\Pp(D>L).
\end{equation*}
Moreover, uniformly for $u\in[0,1]$,
\begin{equation*}
0\leq\phi_L(u)-\phi(u)\leq\Pp(D>L).
\end{equation*}
The latter bound implies convergence of the corresponding maxima of $\phi_L-\Psi$ to that of $\phi-\Psi$. Applying the bounded-degree formula for each fixed $L$, then letting $L\to\infty$ in these bounds, proves the fixed-count limit for arbitrary $D$, including distributions with infinite mean. 

Finally, couple exactly $m$ arrivals and $K\sim\Pois(m)$ arrivals using the same infinite sequence of neighborhoods. Adding or deleting one arrival changes the maximum matching size by at most one. Therefore,
\begin{equation*}
\left|\E[\OPTP(\mathbf N)]-\OPT_m\right|
\leq\E[|K-m|]
\leq\sqrt{\operatorname{Var}(K)}
=\sqrt m.
\end{equation*}
Dividing by $m$ and applying the fixed-count limit proves the claim.
\end{proof}

\bibliographystyle{plain}
\bibliography{bibliography}

\end{document}